\documentclass[11pt, a4paper, leqno]{amsart}
\usepackage{esint} 
\usepackage{graphicx} 
\usepackage[usenames,dvipsnames]{xcolor}
\usepackage{varwidth}
\usepackage[colorlinks=true,linkcolor=Red,citecolor=Green]{hyperref}

\usepackage{amsthm}
\usepackage{amsmath}
\usepackage{amsfonts}
\usepackage{amssymb}
\usepackage{amscd}
\usepackage{mathrsfs}
\usepackage{enumerate}
\usepackage{bm}
\usepackage{dsfont}

\usepackage[utf8]{inputenc}
\usepackage[english]{babel}
\usepackage{anysize}

\usepackage{lipsum}

\newcommand{\R}{\mathbb{R}}
\newcommand{\Op}{\operatorname{Op}}

\newtheorem{prop}{Proposition}
\newtheorem{lemma}{Lemma}

\newtheorem{definition}{Definition}
\newtheorem{corol}{Corollary}
\newtheorem{teor}{Theorem}

\theoremstyle{remark}

\definecolor{darkgreen}{RGB}{0,100,10}

\theoremstyle{remark}
\newtheorem{remark}{Remark}

\title[Convergence of quantum Lindstedt series]{Convergence of quantum Lindstedt series and semiclassical renormalization}
\author{Víctor Arnaiz}

\address{Laboratoire de Mathématiques Jean Leray, Nantes Université. UMR CNRS 6629, 2
	rue de la Houssinière, 44322 Nantes Cedex 03, France.}

\email{victor.arnaiz@univ-nantes.fr}

\begin{document}

\begin{abstract}
In this work we consider the KAM renormalizability problem for small pseudodifferential perturbations of the semiclassical iso\-chro\-nous transport operator with Diophantine frequencies on the torus. Assuming that the symbol of the perturbation is real analytic and  globally bounded, we prove convergence of the quantum Lindstedt series and describe completely the set of semiclassical measures and quantum limits of the renormalized system. Each of these measures is given by symplectic deformation of the Haar measure on an invariant torus for the unperturbed classical system.
\end{abstract}

\maketitle

\section{Introduction and main results}

The present work is concerned with the \textit{renormalization problem} for quantum Hamiltonian systems. Let $\mathbb{T}^d := \R^d/2\pi \mathbb{Z}^d$ be the flat torus, we consider the linear Hamiltonian $\mathcal{L}_\omega : T^*\mathbb{T}^d \to \R$ defined by
\begin{equation}
\label{classical_unperturbed_hamiltonian}
\mathcal{L}_\omega (x,\xi) = \omega \cdot \xi, \quad (x,\xi) \in \mathbb{T}^d \times \R^d \simeq T^*\mathbb{T}^d,
\end{equation}
where the vector of frequencies $\omega$ satisfies the Diophantine condition \eqref{e:diophantine}. The renor\-ma\-lization problem in the classical framework \cite{Elia89,Gall82,Gen96} wonders if, given an open neighborhood $\mathcal{U} \subset \R^d$ of $\xi = 0$, and a small analytic function $ V = V(x,\xi)$ defined on $\mathbb{T}^d \times \mathcal{U}$ and satisfying that $V(x,\xi) = O(\vert \xi \vert^2)$ as $\vert \xi \vert \to 0$, there exists an analytic \textit{counterterm} $R = R( \xi)$ (not depending on the $x$-variable) such that the renormalized Hamiltonian
$$
\mathcal{L}_\omega + V - R
$$
becomes integrable and canonically conjugate to the unperturbed Hamiltonian $\mathcal{L}_\omega$, that is, there exists a canonical transformation $\Phi : \mathbb{T}^d \times \mathcal{U}'' \to \mathbb{T}^d \times \mathcal{U}'$, where $\mathcal{U}'' \subset \mathcal{U}' \subset \mathcal{U}$ are open neighborhoods of zero, such that
\begin{equation}
\label{e:classic_conjugation}
\Phi^* \big( \mathcal{L}_\omega + V - R \big) = \mathcal{L}_\omega.
\end{equation}
This statement, conjectured by Gallavotti \cite{Gall82} and first proven by Eliasson \cite{Elia89} (see also \cite{Gen96}), can be regarded as a control theory theorem. Despite the fact that small perturbations of $\mathcal{L}_\omega$ could generate even ergodic behavior (see \cite{Kat73}) due to degeneracy of the Hessian of $\mathcal{L}_\omega$, this shows that modifying in a suitable way the integrable part of the Hamiltonian, the system remains stable. This also extends a theorem of Rüssmann \cite{Russmann67} which shows convergence of the canonical transformation $\Phi$ provided that $\omega$ satisfies \eqref{e:diophantine} and the (formal) normal form is linear, meaning that all terms in the formal series vanish except the first linear one $\mathcal{L}_\omega$. 

Our main goal is to obtain a quantum analogue of \cite[Thm. A]{Elia89}. We consider families of pseudodifferential perturbations of the iso\-chro\-nous transport operator $-i \hbar \, \omega \cdot \nabla_x$ (where $\hbar \in (0,1]$ is regarded as semiclassical parameter) and study its renormalizability in quantum sense, that is, the existence of suitable integrable counterterms that renormalize the system to make it unitarily conjugate to the unperturbed one. As an application of our result, we will be able to describe completely the semiclassical asymptotics of the renormalized system, obtaining, from any prescribed perturbation, a unique isospectral deformation of $-i \hbar \, \omega \cdot \nabla_x$ for which the quantum Birkhoff normal form becomes \textit{convergent} uniformly as $\hbar \to 0$, and an \textit{exact quantization formula} in the sense of \cite{Pau12,Paul16}, where the authors describe some of the very few examples of systems with convergent quantum Birkhoff normal form so far. This work provides a large new family of KAM systems for which the Birkhoff normal form converges both in classical and quantum sense.

Our motivation comes from previous studies on semiclassical asymptotics for integrable systems (see for instance \cite{ALM:16,ALMCras} for the study of quantum limits of the Dirichlet Laplacian on the disk, \cite{Ar_Mac22,Ar_Mac22B,Ar_Riv18} for semiclassical asymptotics of harmonic oscillators, \cite{An_Mac14,Mac_Riv19} in the case of the Laplacian on the flat torus, or \cite{MaciaZoll,MaciaRiviere16,Mac_Riv18} for Zoll manifols, among many others). A small perturbation of a quantum integrable operator, similarly as in the classical setting, can sometimes dramatically change spectra and propagation phenomena (see \cite{Galkowski18} for an extreme situation in which a small perturbation of an integrable Laplacian becomes unique quantum ergodic). Perturbations of quantum integrable systems for which KAM theory applies are however considerably more subtle and little is known about the precise description of quantum limits for this kind of systems. Most of the works dealing with KAM theory in the semiclassical setting are based on the construction of quasimodes (see \cite{Laz93,Pop00I,Pop00II}), giving rise to semiclassical asymptotics providing very precise estimates of the spectrum distribution. Only very recent results \cite{Gomes18,Gomes_Hassell22} shed some light on the semiclassical asymptotics of sequences of true eigenfunctions of certain KAM systems with discret spectrum, showing absence of quantum ergodicity in general dimension and positive concentration on invariant tori for some KAM systems in dimension two. Our results go in the converse direction; instead of studying directly the spectrum of the perturbed operator, we identify that (resonant) part of the perturbation which generates divergencies in the normal form and outweight it by addition of  integrable counterterms. In \cite{Ar20}, the author considered the case of subprincipal perturbations of the transport operator with Diophantine frequencies via the study of convergence of the quantum Birkhoff normal form by a KAM iterative method. The present work generalizes this result to the case of principal $O(1)$ perturbations and provides a unified approach to the renormalization problem in the quantum semiclassical setting via the study of the Lindstedt series.

Renormalization techniques have been studied by several authors in the context of formal perturbation expansions in quantum field theory, as well as its connection with KAM theory (see for instance \cite{Chan98, Feld92, Gall95, Koch99, Khan88}), adressing in particular the study of convergence of Hamiltonian series expansions arising in the study of quasiperiodic solutions of Hamiltonian systems coming back to the works of Lindstedt and Poincaré \cite{Ar89,Poincare57}.

We also mention that the problem considered here is intimately connected with the reducibility problem for linear quantum Hamiltonian systems (see \cite{Bam_Gre18,Langella19,Procesi19} among others). We emphasize, comparing with these results which hold at quantum ($\hbar = 1$) level, that obtaining our results in full generality, that is, proving the stability of the convergence of the Birkhoff normal form in the semiclassical limit as $\hbar \to 0^+$, requires to estimate the formal series by a careful use of semiclassical pseudodifferential calculus, and not only the algebra of operators. Moreover, we do not restrict ourselves to consider linear or quadratic perturbations (for which the quantum-classical correspondence given by Egorov's theorem becomes exact), allowing the perturbation to belong to a whole space of operators with analytic symbols. This entails some difficulties (see \cite{Pau12} for a similar setting) appearing along the iterative KAM scheme via loss of analiticity. We elucidate and overcome most of these difficulties in this work.

\subsection{Main results.}

We now state the quantum version of the above problem. First of all, we make a strong global assumption. Instead of considering local perturbations of $\mathcal{L}_\omega$ near $\xi = 0$, we consider small analytic perturbations $V(\epsilon,x,\xi)$ globally bounded on the whole phase-space $T^*\mathbb{T}^d \simeq \mathbb{T}^d \times \R^d$. More precisely, we consider, for any $s > 0$, $\epsilon > 0$, the space of real analytic functions on $T^*\mathbb{T}^d$:
\begin{equation}
\label{e:analytic_family}
\mathcal{A}_{s,\epsilon}(T^*\mathbb{T}^d) := \left \{ a \in \mathcal{C}^\omega\big([0,\epsilon) \times T^* \mathbb{T}^d \big) \, : \, \Vert a \Vert_{s,\epsilon} < \infty  \right \},
\end{equation}
with norm $\Vert \cdot \Vert_{s,\epsilon}$ given by 
\begin{equation}
\label{e:norm}
\Vert a \Vert_{s, \epsilon} := \sum_{n=0}^\infty \epsilon^n \Vert a_n \Vert_s,
\end{equation}
where $a_n(x,\xi) = \frac{ \partial_t^n a(t,x,\xi)}{n!} \vert_{t = 0}$, and we define the weighted norm:
\begin{equation}
\label{e:norm_intro}
\Vert a \Vert_s := \int_{\mathcal{Z}^d} \vert \mathcal{F}a(w) \vert e^{s \vert w \vert} \kappa(w),
\end{equation}
where $\mathcal{F} : L^2(T^*\mathbb{T}^d) \to L^2(\mathcal{Z}^d, \kappa)$ is the Fourier transform given by \eqref{e:fourier_transform}. Here we use the conventions of Appendix \ref{s:tools_of_analytic_calculus}, so that $w = (k,\eta) \in \mathcal{Z}^d = \mathbb{Z}^d \times \R^d$ and \eqref{e:norm_intro} is written as a Lebesgue-Stieltjes integral in terms of the measure $\kappa$ defined by \eqref{e:lebesgue_stieltjes}. Similarly we define $\mathcal{A}_{s,\epsilon}(\R^d)$ the space of symbols that do not depend on $x \in \mathbb{T}^d$. If the functions considered do not depend on $\epsilon$, we drop this index in the above definitions. These spaces behave particularly well with respect to the symbolic pseudodifferential calculus. In particular, Calderón-Vaillancourt theorem and precise commutator estimates hold on this family of spaces (see Appendix \ref{s:tools_of_analytic_calculus}).

We now define the semiclassical transport operator
$$
\widehat{L}_{\omega,\hbar} := \Op_\hbar ( \mathcal{L}_\omega) = - i \hbar \, \omega \cdot \nabla,
$$ 
where $\Op_\hbar(\cdot)$ stands for the semiclassical Weyl quantization with semiclassical parameter $\hbar \in (0,1]$ (see Definition \ref{d:semiclassical:quantization}), and we consider perturbations of $\widehat{L}_{\omega,\hbar}$ of the form
$$
\widehat{P}_{\hbar,t}(V) := \widehat{L}_{\omega,\hbar} + t \Op_\hbar(V), \quad 0 \leq t \leq \epsilon,
$$
where $V \in \mathcal{A}_s(T^*\mathbb{T}^d)$ is a prescribed real analytic function. In the sequel, we need to impose a strong non-resonant condition on the vector of frequencies $\omega \in \R^d$; we assume that $\omega$ satisfies the following Diophantine condition: there exist $\varsigma > 0$ and $\gamma > d-1$ such that
\begin{equation}
\label{e:diophantine}
\vert \omega \cdot k \vert \geq \frac{\varsigma}{\vert k \vert^\gamma}, \quad k \in \mathbb{Z}^d \setminus \{0 \}.
\end{equation}

In order to renormalize the operator $\widehat{P}_{\hbar,\epsilon}(V)$, we require the addition of an integrable counterterm. Let $R \in \mathcal{A}_{s,\epsilon}(\R^d)$, we set:
 $$
 \widehat{P}_{\hbar,t}(V,R) :=  \widehat{L}_{\omega,\hbar} + \Op_\hbar(tV - R(t)), \quad 0 \leq t \leq \epsilon.
 $$
 The main result of this work is the following global semiclassical version of the classical renormalization problem described above:

\begin{teor}
\label{c:1}
Let $\omega \in \R^d$ satisfy \eqref{e:diophantine}. Then, given $s_0 > 0$ and  $V \in \mathcal{A}_{s_0}(T^*\mathbb{T}^d)$, there exist $\epsilon = \epsilon(V,\omega) > 0$, $0 < s \leq s_0$, a counterterm $R = R_\hbar \in \mathcal{A}_{s,\epsilon}(\R^d)$, uniformly bounded for $\hbar \in [0,1]$, and a unitary operator $t \mapsto \mathcal{U}_\hbar(t)$ on $L^2(\mathbb{T}^d)$, depending analytically on $t \in [0,\epsilon]$ such that, for every $ \hbar \in (0,1]$ and $t \in [0,\epsilon]$,
\begin{equation}
\label{e:quantum_renormalization}
\mathcal{U}_\hbar(t)^* \widehat{P}_{\hbar,t}(V,R) \, \mathcal{U}_\hbar(t) = \widehat{L}_{\omega,\hbar}.
\end{equation}
\end{teor}
This shows in particular that the $L^2(\mathbb{T}^d)$-spectrum of the operator $ \widehat{P}_{\hbar,t}(V,R)$ coincides with that of the unperturbed operator $\widehat{L}_{\omega,\hbar}$. In other words, the renormalization procedure generates an isospectral deformation of $\widehat{L}_{\omega,\hbar}$ from any prescribed perturbation $V \in \mathcal{A}_{s_0}(T^*\mathbb{T}^d)$. In particular, the spectrum all along this family is pure-point, as it is the spectrum of $\widehat{L}_{\omega,\hbar}$, that is: there exists and orthonormal basis of $L^2(\mathbb{T}^d)$ consisting of eigenfunctions. Recall that the point spectrum of $\widehat{L}_{\omega,\hbar}$ is given by:
$$
\operatorname{Sp}^{\textnormal{p}}_{L^2(\mathbb{T}^d)} (\widehat{L}_{\omega,\hbar}) = \{ \hbar  \, \omega \cdot k \, : \, k \in \mathbb{Z}^d \}. 
$$
 
We next aim at describing the semiclassical asymptotics of the renormalized system $\widehat{P}_{\hbar,t}(V,R)$. Let us define the set of semiclassical measures of the operator $  \widehat{P}_{\hbar,t}(V,R)$ as the set of probabilty measures supported on $ (\mathcal{L}_\omega + tV - R(t))^{-1}(1) \subset T^* \mathbb{T}^d $ that are weak-$\star$ limits of sequences of Wigner distributions associated with normalized sequences of eigenfunctions $(\Psi_\hbar, \lambda_\hbar)$ satisfying
\begin{equation}
\label{e:eigenfunction_sequence}
 \widehat{P}_{\hbar,t}(V,R) \Psi_\hbar = \lambda_\hbar \Psi_\hbar, \quad \Vert \Psi_\hbar \Vert_{L^2(\mathbb{T}^d)} = 1, \quad \lambda_\hbar \to 1.
\end{equation}
That is, $\mu$ is a semiclassical measure (see \cite{Ger90}) of $\widehat{P}_{\hbar,t}(V,R)$ for the sequence $(\Psi_\hbar, \lambda_\hbar)$ if, modulo a subsequence, for every $a \in \mathcal{C}_c^\infty(T^* \mathbb{T}^d)$,
 $$
 \lim_{\hbar \to 0^+} \big \langle \Op_\hbar(a) \Psi_\hbar, \Psi_\hbar \big \rangle_{L^2(\mathbb{T}^d)}  = \int_{T^* \mathbb{T}^d} a \, d\mu.
 $$
Notice that semiclassical measures are defined on the phase-space $T^*\mathbb{T}^d$. One can project these measures onto the position variable by testing the sequence against symbols $a$ depending on the $x$ variable only. These projections are usually called \textit{quantum limits}. More precisely, a quantum limit $\nu$ for a sequence $(\Psi_\hbar, \lambda_\hbar)$ satisfying \eqref{e:eigenfunction_sequence} is a probability measure on $\mathbb{T}^d$ such that, for any $b \in \mathcal{C}(\mathbb{T}^d)$,
$$
\lim_{\hbar \to 0} \int_{\mathbb{T}^d} b(x) \vert \Psi_\hbar(x) \vert^2  dx = \int_{\mathbb{T}^d} b(x) d\nu(x).
$$
In particular, any quantum limit is obtained by projection of a semiclassical measure, so that
$$
\nu(x) = \int_{\R^d} \mu(x,d\xi).
$$
\begin{teor}
\label{c:2}
Let $\omega \in \R^d$ satisfy \eqref{e:diophantine}, and let $s_0  > 0$. Given $V \in \mathcal{A}_{s_0}(T^*\mathbb{T}^d)$, let $\epsilon = \epsilon(V, \omega) > 0$ be given by Theorem \ref{c:1}. Then there exists a symplectomorphism $t \mapsto \Phi_t : T^*\mathbb{T}^d \to T^*\mathbb{T}^d$ depending analytically on $t \in [0,\epsilon)$ such that 
$$
\Vert \operatorname{Id} - \Phi_t \Vert_{s} \leq C \epsilon, \quad C = C(V,\omega) > 0,
$$ 
and the set of semiclassical measures of the operator $\widehat{P}_{\hbar,t}(V,R)$ is precisely:
$$
\left \{ \mu = (\Phi_t)_* \, \delta_{\mathbb{T}^d \times \{ \xi_0 \} }  \, : \, \mathbb{T}^d \times \{\xi_0 \} \subset \mathcal{L}_\omega^{-1}(1) \right \},
$$
where $\delta_{\mathbb{T}^d \times \{ \xi_0 \}}$ denotes the uniform probability measure (normalized Haar measure) on $\mathbb{T}^d \times \{ \xi_0 \}$. 
\end{teor}

\begin{remark}
If $V \equiv 0$, the result is trivial (see \cite[Prop. 1]{Ar20}). 
\end{remark}

\begin{remark}
The canonical transformation $\Phi_t$ necessarily concides with the one coming from the classical problem. Moreover, $\mathcal{U}_\hbar(t)$ is a Fourier integral operator quantizing this symplectomorphism which satisfies Egorov's theorem:
$$
\mathcal{U}_\hbar(t)^* \, \Op_\hbar(a) \, \mathcal{U}_\hbar(t) = \Op_\hbar(a \circ \Phi_t) + O_{\mathcal{L}(L^2(\mathbb{T}^d))}(\hbar).
$$
\end{remark}

\begin{remark}
If we allow $\epsilon$ to depend on $\hbar$ and to be of the form $\epsilon_\hbar = \epsilon \hbar$ (subprincipal perturbation), then Theorems \ref{c:1} and \ref{c:2} have essentially been proven in \cite[Thm. 2]{Ar20}. In this case, moreover, the symplectormorphism $\Phi$ of Theorem \ref{c:2} is actually $\Phi = \operatorname{Id}$. The proof is based on a normal-form approach (KAM iterative method) at operator level. However, if the perturbation is of order $\epsilon$ (principal level), then the normal form has to be estimated necessarily using the semiclassical calculus, which brings out many new issues. To the best of the author knowledge, no other proofs of the classical renormalization problem exist out from the study of convergence of Lindstedt series (\cite{Elia89,Gen96} and the references therein). This justifies the use of this method ahead of other KAM methods such as Nash-Moser theorem or other KAM iterative schemes. The problem of renormalization of the semiclassical operator $\widehat{L}_{\omega,\hbar} + \epsilon \Op_\hbar(V)$ (uniformly in the semiclassical parameter $\hbar > 0$) is a major challenge, since the relative scale of the perturbation with respect to the principal transport operator is much larger, so it is not sufficient to work at operator level and we are forced to deal with the semiclassical pseudodifferential calculus. On the other hand, in the particular case in which $\hbar \equiv 1$ is fixed, then \cite[Thm. 2]{Ar20} is enough to renormalize the operator $\widehat{L}_{\omega,1} + \epsilon \Op_1(V)$ and it is most likely that other KAM methods can be adapted to show convergence of the Birkhoff normal form.
\end{remark}

\begin{remark}
If $V$ has the particular form $V = V(x,\omega \cdot \xi)$, and
$$
\sum_{k \in \mathbb{Z}^d}  \int_{\R} \vert \mathcal{F}{V}(k,\tau) \vert \exp \big( s_0(\vert k \vert + \vert \tau \vert) \big) d\tau <  \infty,
$$ 
for some $s_0 > 0$, where $\mathcal{F}$ denotes here the Fourier transform in $(x,y) \in \mathbb{T}^d \times \R$, then the normal form converges in quantum sense (without necessity of renormalization), that is, there exists $0 < \epsilon = \epsilon(V,\omega)$ and $\mathcal{U}_\hbar$ so that
$$
\mathcal{U}_\hbar^* \, \big( \widehat{L}_{\omega,\hbar} + \epsilon \Op_\hbar( V ) \big) \, \mathcal{U}_\hbar = \widehat{L}_{\omega,\hbar} + \sum_{n=1}^\infty  \Op_\hbar(R_{n,\hbar}),
$$
where $R_{n,\hbar} = R_{n,\hbar}(\epsilon, \xi)$, and $R_\hbar = \sum_{n=1}^\infty R_{n,\hbar}$ satisfies $\Vert R_\hbar \Vert_{s} < \infty$ for some $0 < s \leq s_0$. This is shown in \cite{Pau12} (see also \cite{Paul16} for a more general result in this direction), where the authors prove convergence of the quantum normal form at semiclassical level, working in the family of symbols \eqref{e:analytic_family}. From this result, one can prove easily Theorem \ref{c:2} (for the original operator $\widehat{L}_{\omega,\hbar} + \epsilon \Op_\hbar(V)$) also in this case.
\end{remark}

Notice that the classical result \cite{Elia89}, \cite{Gen96} has been only proven locally in a neighborhood of $\xi = 0$, while our quantum extension is global in $T^*\mathbb{T}^d$. We obtain as a byproduct of our study the following global classical theorem:

\begin{teor}
\label{t:classic_main_theorem}
Let $\omega \in \R^d$ satisfy \eqref{e:diophantine}, and let $s_0  > 0$. Then, given $V \in \mathcal{A}_{s_0}(T^*\mathbb{T}^d)$ there exist $0 < s \leq s_0$, $0 < \epsilon = \epsilon(V,\omega)$, a counterterm $R \in \mathcal{A}_{s,\epsilon}(\R^d)$ and a symplectomorphism $t \mapsto \Phi_t : T^*\mathbb{T}^d \to T^* \mathbb{T}^d$, depending analytically on $t \in [0,\epsilon)$ such that
$$
\Vert \operatorname{Id} - \Phi_t \Vert_{s} \leq C \epsilon, \quad \Phi_t^* \big( \mathcal{L}_\omega +  t V - R(t) \big)  = \mathcal{L}_\omega.
$$ 
\end{teor}

\begin{remark}
The counterterms $R_\hbar$ and $R$ obtained respectively for the quantum and the classical problem satisfy $R_\hbar = R + O_{\mathcal{A}_{s,\epsilon}(\R^d)}(\hbar)$.
\end{remark}

From the proof of Theorem \ref{c:1} one can easily obtain some other (weaker) versions of the renormalization problem in the quantum setting. In particular, if $\hbar = 1$ is fixed or the size of the perturbation is of subprincipal type ($\epsilon_\hbar = \epsilon \hbar$), one can show Theorem \ref{c:1} even relaxing the analyticity hypothesis in the $\xi$ variable, and requering only that $\xi \mapsto V(x,\xi)$ is bounded together with all its derivatives. By the Lindstedt series approach one can show \cite[Thm. 2]{Ar20} also in the Sjöstrand-class $\mathcal{A}_{W,s}(T^*\mathbb{T}^d)$ considered in \cite[Def. 3]{Ar20}. Also in the case of linear perturbations of the form $V(x,\xi) = v(x) \cdot \xi$, \cite[Corol. 1]{Ar20} can be reproved using the Lindstedt series approach, showing moreover analyticity of the renormalization with respect to the size of the perturbation (see Remark \ref{r:simplification}).


\subsection*{Acknowledgments}

The author would like to warmly thank Fabricio Macià, Gabriel Rivière, Colin Guillarmou, Chenmin Sun, Benoît Grébert, and Georgi Popov for usefull discussions on this and related problems, and to Alberto Maspero and Massimiliano Berti for their kindly invitation to the SISSA (Trieste) in 2019, where this work was conceived. This research has been supported by ANR project Aléatoire, Dynamique et Spectre. The author is also partially supported by projects: MTM2017-85934-C3-3-P and PID2021-124195NB-C31 (MINECO, Spain).

\medskip

\section{The Lindstedt series}
\label{s:lindstedt_series}

We start our study from the formal conjugation problem \eqref{e:main_equation_for_lindstedt} which brings out the Lindstedt series. The analysis of this series will be the key ingredient for the proof of Theorem \ref{c:1}. We  revisit some parts of the work of Eliasson \cite{Elia89}  (see also \cite{Eliasson90}, \cite{Eliasson96}) and adapt it to the quantum setting. Some steps in the proofs are here detailed or modified to a more convenient exposition in line with the present work. The main difference with respect to this series of works is that, instead of looking for the canonical transformation $\Phi_t$ as a general symplectic isomorphism homotopic to the identity, we restrict ourselves to consider time-dependent Hamiltonian flows $\Phi_t = \Phi_t^H$ so that $\Phi_t^H \vert_{t = 0} = \operatorname{Id}$. We show that this reduction does not lead to a loss of generality (the solution is formally unique, see \cite{Gall82}) and turns out to be very usefull to quantize the problem. Indeed we show formal solvability and convergence of the Birkhoff normal form in a unified way, so that the classical result is derived as a byproduct from the quantum one. Moreover, we show that the Lindstedt series for the formal expansion of $H$ can be linked with the one obtained by Eliasson for  the canonical transformation $\Phi_t$.

\subsection{Time-dependent Hamiltonian flows}
\label{s:hamiltonians}

In this section we describe the evolution of quantum Hamiltonian systems with operators having symbols in the analytic spaces \eqref{e:analytic_family}. Let $H \in \mathcal{A}_{s,\epsilon}(T^*\mathbb{T}^d)$ be a time-dependent classical Hamiltonian. We consider the quantum initial value problem:
\begin{equation}
\label{e:evolution_problem}
\left \lbrace \begin{array}{l}
\hbar D_t U_H(t) + \Op_\hbar(H(t)) U_H(t) = 0, \\[0.2cm]
 U_H(0) = \operatorname{Id}.
\end{array} \right.
\end{equation} 
The existence and uniqueness of the solution to \eqref{e:evolution_problem} follows by \cite[Thm. 10.1]{Zw12} since the operator $\Op_\hbar(H(t))$ is selfadjoint, depends smoothly on $t$, and by Calderón-Vaillancourt theorem (see Lemma \ref{l:analytic_calderon_vaillancourt_torus})  is uniformly bounded on $L^2(\mathbb{T}^d)$ for $t \in [0, \epsilon]$. Moreover, if $U_H(t)$ satisfies \eqref{e:evolution_problem}, then $U_H(t)^*$ solves the adjoint evolution problem
\begin{equation}
\label{e:adjoint_evolution_equation}
\left \lbrace \begin{array}{l}
\hbar D_t U_H(t)^* - U_H(t)^*  \Op_\hbar(H(t)) = 0, \\[0.2cm]
 U(0)^* = \operatorname{Id}.
\end{array} \right.
\end{equation}
Then, for any semiclassical operator $\Op_\hbar(a)$ with $a \in \mathcal{A}_s(T^*\mathbb{T}^d)$, we have:
\begin{equation}
\label{e:general_conjugation_equation}
\hbar D_t \big( U_H(t)^* \Op_\hbar(a) U_H(t) \big) = U_H(t)^*[ \Op_\hbar (H), \Op_\hbar(a)]  U_H(t).
\end{equation}
This conjugation equation can be read also at symbol level in the family of analytic spaces \eqref{e:analytic_family}, due to the particularly good behavior of the symbolic calculus on these spaces. Indeed, by Lemma \ref{l:commutator_loss}, if $a,b \in \mathcal{A}_s(T^*\mathbb{T}^d)$, then 
$$
\frac{i}{\hbar} [\Op_h(a), \Op_h(b)]_\hbar :=  \Op_h([a,b]_h)
$$ 
with $[a,b]_\hbar \in \mathcal{A}_{s-\sigma}(T^*\mathbb{T}^d)$ for every $0 < \sigma < s$. Moreover, by Lemma \ref{l:propagator_symbolic_lemma}, for every $0 < \sigma < s$, if $\epsilon \Vert H \Vert_{s,\epsilon}$ is sufficiently small with respect to $\sigma$, then for every $a \in \mathcal{A}_s(T^*\mathbb{T}^d)$, 
$$
U_H(t)^* \Op_\hbar(a) \, U_H(t) = \Op_h( \Psi_t^H(a)),
$$ 
with $\Psi_t^H(a) \in \mathcal{A}_{s-\sigma}(T^*\mathbb{T}^d)$ for every $0 \leq t \leq \epsilon$. This means that, at symbol level, \eqref{e:general_conjugation_equation} reduces to the equation:
\begin{equation}
\label{e:symbolic_differential_equation}
\frac{d}{dt} \Psi_t^H(a) = \Psi_t^H\big( [H(t), a]_\hbar \big).
\end{equation}
Similarly, denoting $U_H(t) \Op_\hbar(a) U_H(t)^* = \Op_\hbar( (\Psi_t^H)^{-1}(a))$, we get:
\begin{equation}
\label{e:inverse_symbol_equation}
\frac{d}{dt} (\Psi_t^H)^{-1}(a) = - [H(t), (\Psi_t^H)^{-1}(a)]_\hbar.
\end{equation}
Notice that, by the symbolic calculus (Weyl quantization) and Egorov's theorem,
$$
[a,b]_\hbar =  \{ a, b \} + O(\hbar^2); \quad \Psi_t^H(a) = \Phi_t^H(a) + O(\hbar),
$$
where $\Phi_t^H$ denotes the (time-dependent) classic flow generated by $H$. Since we work in the analytic framework, we can estimate globally these objects via loss of analyticity (see \cite{Ar20,Pau12} and Appendix \ref{s:tools_of_analytic_calculus}).

We now provide some explicit formulas for the flow $(\Psi_t^{-H})^{-1}$, which will be particularly usefull in next section, regarding the conjugation equation \eqref{e:conjugation_equation} below. Using \eqref{e:inverse_symbol_equation}, we observe that
\begin{equation}
\label{e:key_flow}
\frac{d}{dt} (\Psi_t^{-H})^{-1}(a) = [H , (\Psi_t^{-H})^{-1}(a)]_\hbar.
\end{equation}
This identity allows us, after expanding formally
$$
H(t) = \sum_{n=1}^\infty t^{n-1} H_n, \quad (\Psi_t^{-H})^{-1}(a) = \sum_{n=0}^\infty t^n \psi_n^{-1}(a),
$$
to find $\psi_0^{-1}(a) = 0$ and, for $n \geq 1$, the recursive relation
$$
\psi_n^{-1}(a) = \frac{1}{n} \sum_{j=0}^{n-1} [H_{n-j}, \psi_j^{-1}(a)]_\hbar.
$$
This provides closed formulas for the coefficients $\psi_n^{-1}(a)$ for $n \geq 1$. We precisely have:
\begin{equation}
\label{e:general_coefficient}
\psi_n^{-1}(a) = \sum_{j=1}^n \sum_{k_1+ \cdots + k_j = n} \mathbf{c}_{k_1,\ldots,k_j}[ H_{k_j}, \cdots ,[ H_{k_1}, a]_\hbar \cdots ]_\hbar,
\end{equation}
where
\begin{equation}
\label{e:coefficients}
\mathbf{c}_{k_1,\ldots, k_j} := \frac{1}{k_1 + \cdots + k_j} \cdot \frac{1}{k_1 + \cdots + k_{j-1}} \cdots \frac{1}{k_1}.
\end{equation}

We finally show some elementary combinatorial lemmas regarding the coefficients $\mathbf{c}_{k_1,\ldots, k_j}$ defined above, which will be usefull in the sequel.
\begin{lemma}
\label{l:break_order}
Let $k_1, \ldots, k_j \in \mathbb{N}$. Let $\pi_j$ be the group of permutations of $j$ elements. Then:
$$
\sum_{\sigma \in \pi_j} \mathbf{c}_{\sigma(k_1, \ldots, k_j)} = \frac{1}{k_1} \cdots \frac{1}{k_j}.
$$
\end{lemma}
\begin{proof}
The proof is an easy induction in $j$. The case $j = 1$ is trivial. Let $j - 1 \geq 1$. For any $i =1, \ldots j$, we write $\mathbf{k}^i = (k_1, \ldots, k_{i-1}, k_{i+1}, \ldots k_j) \in \mathbb{N}^{j-1}$ the vector obtained by removing the term $k_i$. Now we use the induction hypothesis to get:
\begin{align*}
\sum_{\sigma \in \pi_j} \mathbf{c}_{\sigma(k_1, \ldots, k_j)} & = \sum_{i=1}^j \sum_{\sigma' \in \pi_{j-1}} \mathbf{c}_{\sigma'(\mathbf{k}^i),i}  \\[0.2cm]
 & = \frac{1}{k_1 + \cdots + k_j} \sum_{i=1}^j \sum_{\sigma' \in \pi_{j-1}} \mathbf{c}_{\sigma'(\mathbf{k}^i)} \\[0.2cm]
 & = \frac{1}{k_1 + \cdots + k_j} \sum_{i=1}^j \frac{1}{k_1 \cdots k_{i-1} k_{i+1} \cdots k_j} \\[0.2cm]
 & = \frac{1}{k_1 \cdots k _j}.
\end{align*}
\end{proof}

\begin{lemma}
\label{l:ordered_permutations}
Let $k_1, \ldots, k_j \in \mathbb{N}$. Let $r \in \{1, \ldots, j \}$, $I^r_1 = \{1, \ldots, r \}$, and $I^r_2 = \{r+1, \ldots, j \}$. Let $\mathcal{O}(r,j)$ be the subset of permutations $\sigma_* \in \pi_j$ that verifies the following property: for every $j = 1,2$, if $i,i' \in I^r_j$ with $i < i'$, then $\sigma_*(i) < \sigma_*(i')$. Then:
$$
\sum_{\sigma_* \in \mathcal{O}(j,r)} \mathbf{c}_{\sigma_*(k_1, \ldots, k_j)} = \mathbf{c}_{k_1, \ldots, k_r} \mathbf{c}_{k_{r+1}, \ldots, k_j}.
$$
\end{lemma}
\begin{proof}
We proceed by induction. For $j = 1$ the claim is trivial. Assume that the claim holds for $j -1 \geq 1$. The case $r = j$ is also trivial. Assume that $r \leq j-1$. Using the induction hypothesis we get:
\begin{align*}
\sum_{\sigma_* \in \mathcal{O}(j,r)} \mathbf{c}_{\sigma_*(k_1, \ldots, k_j)} & = \sum_{\sigma_* \in \mathcal{O}(j-1,r)} \mathbf{c}_{\sigma_*(k_1, \ldots, k_{j-1}),k_j} + \sum_{\sigma_* \in \mathcal{O}(j-1,r-1)} \mathbf{c}_{\sigma_*( \mathbf{k}^r),k_r} \\[0.2cm]
 & = \frac{1}{k_1 + \cdots + k_j} \Big( \sum_{\sigma_* \in \mathcal{O}(j-1,r)} \mathbf{c}_{\sigma_*(k_1, \ldots, k_{j-1})} + \sum_{\sigma_* \in \mathcal{O}(j-1,r-1)} \mathbf{c}_{\sigma_*( \mathbf{k}^r)} \Big) \\[0.2cm]
 & = \frac{1}{k_1 + \cdots + k_j}  \Big( \mathbf{c}_{k_1, \ldots, k_r} \mathbf{c}_{k_{r+1}, \ldots, k_{j-1}} + \mathbf{c}_{k_1, \ldots, k_{r-1}} \mathbf{c}_{k_{r+1}, \ldots, k_j} \Big) \\[0.2cm]
 & = \mathbf{c}_{k_1, \ldots, k_r} \mathbf{c}_{k_{r+1}, \ldots, k_j}.
\end{align*}
\end{proof}

\begin{lemma}
\label{l:jacobi_coefficients}
Let $l_1, l^0_1, \ldots, l^0_i \in \mathbb{N}$. Then:
$$
\frac{1}{l_1} \mathbf{c}_{l_1^0, \ldots, l_i^0} = \mathbf{c}_{l_1^0, \ldots, l_i^0, l_1} + \frac{1}{l_1} \mathbf{c}_{l_1^0, \ldots, l_i^0 + l_1}.
$$
\end{lemma}
\begin{proof}
The proof follows inmediately from \eqref{e:coefficients}.
\end{proof}
\begin{lemma}
\label{l:transposition_lemma}
Let $l_1, \ldots, l_s, l^0_1, \ldots, l^0_i \in \mathbb{N}$. Then:
\begin{align*}
\mathbf{c}_{l_1, \ldots, l_s} \mathbf{c}_{l^0_1, \ldots, l^0_i} & = \mathbf{c}_{l^0_1, \ldots, l^0_i, l_1, \ldots, l_s} + \mathbf{c}_{l_1} \mathbf{c}_{l_1^0, \ldots, l^0_i + l_1, l_2, \ldots, l_s} \\[0.2cm]
 & \quad + \cdots  + \mathbf{c}_{l_1, \ldots, l_s} \mathbf{c}_{l^0_1, \ldots, l^0_i + l_1 + \cdots + l_s}.
\end{align*}
\end{lemma}
\begin{proof}
The case $s = 1$ follows by Lemma \ref{l:jacobi_coefficients}. For the general case, notice that:
\begin{align*}
\mathbf{c}_{l_1, \ldots, l_s} \mathbf{c}_{l^0_1, \ldots, l^0_i} - \mathbf{c}_{l_1, \ldots, l_s} \mathbf{c}_{l^0_1, \ldots, l^0_i + l_1 + \cdots + l_s} & \\[0.2cm]
 & \hspace*{-3cm}  = \mathbf{c}_{l_1, \ldots, l_s} \mathbf{c}_{l^0_1, \ldots, l^0_{i-1}} \left( \frac{1}{l^0_1 + \cdots + l^0_i} - \frac{1}{l^0_1+ \cdots + l^0_i + l_1 + \cdots + l_s}\right) \\[0.2cm]
 & \hspace*{-3cm} =  \mathbf{c}_{l_1, \ldots, l_{s-1}}  \mathbf{c}_{l^0_1, \ldots, l^0_{i}, l_1 + \cdots + l_s}.
\end{align*}
Similarly,
\begin{align*}
\mathbf{c}_{l_1, \ldots, l_{s-1}}  \mathbf{c}_{l^0_1, \ldots, l^0_{i},l_1 + \cdots + l_s} -  \mathbf{c}_{l_1, \ldots, l_{s-1}} \mathbf{c}_{l^0_1, \ldots, l^0_i + l_1 + \cdots + l_{s-1}, l_s} & \\[0.2cm]
 & \hspace*{-7.5cm} =  \frac{\mathbf{c}_{l_1, \ldots, l_{s-1}} \mathbf{c}_{l^0_1, \ldots, l^0_{i-1}}}{l^0_1+ \cdots + l^0_i + l_1 + \cdots + l_s} \left( \frac{1}{l^0_1 + \cdots + l^0_i} - \frac{1}{l^0_1 + \cdots + l^0_i + l_1 + \cdots + l_{s-1}} \right) \\[0.2cm]
 & \hspace*{-7.5cm} =  \mathbf{c}_{l_1, \ldots, l_{s-2} } \mathbf{c}_{l^0_1, \ldots, l^0_{i},l_1 + \cdots + l_{s-1},l_s} .
\end{align*}
Iterating this process we obtain the claim. Notice that in the last iteration we find
$$
 \mathbf{c}_{l^0_1, \ldots, l^0_i,l_1, \ldots, l_s} - \mathbf{c}_{l^0_1, \ldots, l^0_i, l_1, \ldots, l_s} = 0.
$$
\end{proof}

\subsection{Formal conjugation and cohomological equations}
In this section, we adress the following conjugation problem. Let $V \in \mathcal{A}_{s_0}(T^*\mathbb{T}^d)$, we consider the non-linear equation:
\begin{equation}
\label{e:main_equation_for_lindstedt}
\frac{i}{\hbar} [ \widehat{L}_{\omega,\hbar}, \Op_\hbar(H(t))] = U_{-H}(t) \Op_\hbar(  V - R'(t)) U_{-H}(t)^*, \quad t \in [0, \epsilon],
\end{equation}
with unknowns $H(t)$ and $R'(t)$. Our goal is to solve this equation for $H \in \mathcal{A}_{s,\epsilon}(T^*\mathbb{T}^d)$ and $R' \in \mathcal{A}_{s,\epsilon}(\R^d)$,  uniformly in $\hbar \in (0,1]$, for certain $0 < s < s_0$ and $\epsilon > 0$ sufficiently small. Notice that, since $\mathcal{L}_\omega$ is a linear symbol, one has the exact commutation relation $[\mathcal{L}_\omega, H(t) ]_\hbar = \{\mathcal{L}_\omega , H(t)\}$; then equation \eqref{e:main_equation_for_lindstedt}  at symbol level reads:
\begin{equation}
\label{e:conjugation_equation}
\{\mathcal{L}_\omega, H(t) \} = (\Psi_t^{-H})^{-1}\big( V - R'(t) \big).
\end{equation}
We next expand formally $H(t) = \sum_{n=1}^\infty t^{n-1} H_n$ and $R'(t) = \sum_{n=1}^\infty t^{n-1} R'_n$ in powers of $t$ and find the first cohomological equations by identifying the terms of same order in $t$:
\begin{align}
\label{e:first_cohomological_equation}
\{ \mathcal{L}_\omega, H_1 \} & = V - R'_1, \\[0.2cm]
\{ \mathcal{L}_\omega, H_2 \} & = [H_1, V - R'_1]_\hbar - R'_2, \\[0.2cm]
\{ \mathcal{L}_\omega, H_3 \} & = \frac{1}{2}[H_1, [H_1, V-R'_1]_\hbar]_\hbar + \frac{1}{2}[H_2, V-R'_1]_\hbar - [ H_1, R'_2]_\hbar - R'_3,
\end{align}
and more generally, we get the $n$-th cohomological equation:
\begin{align}
\label{e:cohomological_general_equation}
\{ \mathcal{L}_\omega , H_n \} & = \sum_{j=1}^{n-1} \sum_{k_1 + \cdots + k_j = n} \mathbf{c}_{k_1,\ldots, k_j} [H_{k_j}, \ldots, [ H_{k_1}, V - R'_1]_\hbar \cdots ]_\hbar  \\[0.2cm]
 & \quad - \sum_{m = 2}^{n-1} \sum_{j=1}^{n-m} \sum_{l_1 + \cdots + l_j = n-m} \mathbf{c}_{l_1, \ldots, l_j} [H_{l_j}, \ldots, [H_{l_1}, R'_m]_\hbar \cdots]_\hbar - R'_n \notag,
\end{align}
where the coefficients $\mathbf{c}_{k_1,\ldots,k_j}$ have been defined in \eqref{e:coefficients}.

Each of these cohomological equations can be solved using Lemma \ref{l:solution_cohomological_equation}. Moreover, we will see that the formal solution of $H_n$ and $R'_n$ can be written by recursive formulas in terms of a diagrammatic tree structure. Before going beyond, we revisit the theory of diagrammatic trees introduced by Eliasson \cite{Elia89,Eliasson96}.

\subsection{Index sets and tree structures} Let $n \in \mathbb{N}$, we consider the set $\Delta(n)$ of mappings $\delta : \{ 1, \ldots, n \} \to \mathbb{N}$ such that, denoting $\delta(i) = \delta_i$,
$$
\sum_{j \leq i \leq n} \delta_i \geq n - j + 1, \quad \text{if } 1 < j \leq n; \quad \sum_{1 \leq i \leq n} \delta_i = n -1.
$$
For notational purposes, in the sequel we identify $\delta$ with the vector $(\delta_1,\ldots, \delta_n)$. Given any subset $A \subset \mathbb{N}$ with $\# A = n$, we can define anologously $\delta: A \to \mathbb{N}$. 

\begin{lemma}
\label{l:counting_trees}
$\# \Delta(n) \leq 4^n$.
\end{lemma}

\begin{proof}
The proof is based on the simple combinatorial estimate:
$$
\# \Delta(n) \leq N(n,n),
$$ 
where
$$
N(n,j) = \#\{ (k_1, \ldots, k_j) \in \mathbb{N}_0^j \, : \, k_1 + \cdots + k_j = n \}.
$$
Indeed, the more general estimate $N(n,j) \leq 2^{n+j}$ follows by an easy induction on $j$. 
\end{proof}

\begin{remark}
Actually $\# \Delta(n)$ is given precisely by the Catalan number 
$$
\Delta(n) = \frac{1}{n} {2n-2 \choose n-1 },
$$
see \cite[Chpt. 2.7]{Goulden_Jackson}.
\end{remark}

\begin{definition}
A simple index set is a finite subset $A \subset \mathbb{Z}$ together with a map $\delta \in \Delta(n)$, $n = \#A$. An index set is a disjoint union of simple index sets.
\end{definition}

Simple index sets are in one-to-one correspence with \textit{rooted planar trees}. To show this correspondence, we first recall the notion of tree structure on $A \subset \mathbb{N}$. A tree structure $\mathcal{T} = (A,\prec)$ is defined by a partial ordering $\prec$  such that $A$ has a unique maximal element and such that $\{ d : c \preceq d \}$ is totally ordered for each $c \in A$ (see Figure \ref{f:figure_1}). Let us denote by $\mathfrak{T}$ the family of tree structures. Two points $c,d$ in such a tree are said to be unrelated if neither $c \preceq d$ nor $d \preceq c$. $d$ is said to be a predecessor of $c$, and $c$ a succesor of $d$, if $d \prec c$, and they are said to be immediate if for no $e$ it holds that $d \prec e \prec c$. We denote $A(c) = \{ e \in A \, : \, e \preceq c \}$ and
$$
A(B) = \bigcup_{c \in B } A(c).
$$
We also set:
\begin{align*}
[a,b] & := \{ c \in A \, : \, a \preceq c \preceq b \},
\end{align*}
and similarly we define $[a,b[$, $]a,b]$, and $]a,b[$.

\begin{center}
\begin{figure}[h]
\includegraphics[scale=0.5]{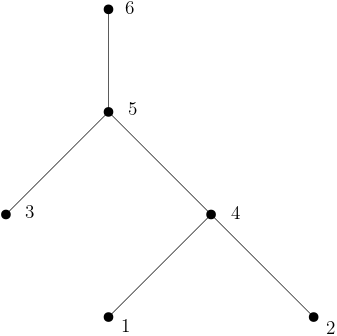}
\caption{Tree diagram with $n = 6$ and $\delta = (0,0,0,2,2,1 )$. We label the tree from the top to the bottom and from the right to the left.}
\label{f:figure_1}
\end{figure}
\end{center} 

We say that two tree structures $\mathcal{T} = (A,\prec)$ and $\mathcal{T}' = (A',\prec')$ define the same rooted planar tree, and denote it by $\mathcal{T} \sim \mathcal{T}'$, if there exists a bijection $\beta : A \to A'$ such that
$$
a \prec b \Leftrightarrow \beta(a) \prec' \beta(b),
$$
and moreover, for any $a \in A$, let  $a_1 < \cdots < a_r$ be the immediate predecessors of $a$, then $\beta(a_1) < \cdots < \beta(a_r)$ are the immediate predecessors of $\beta(a)$. We call $\mathfrak{T}' := \mathfrak{T}/\sim$ the family of rooted planar trees and denote by $[\mathcal{T}]$ the equivalence class of the tree structure $\mathcal{T}$.

\begin{definition}
Let $n \geq 1$ and let $\mathcal{T} = (A,\prec)$ be a tree structure with $n = \# A$. We say that two points $a,b \in A$ such that $a \preceq b$ are at distance $l$ if the set $[a,b[$ has exactly $l$ elements (in particular $a \in A$ is at distance zero from itself). We define the \textit{diameter} $\mathbf{d}(\mathcal{T})$ as:
$$
\mathbf{d}(\mathcal{T}) := \max \{ \# [a,b[  \, : \, a \prec b, \quad a,b \in A \}.
$$
In particular, if $n = 1$, then $\mathbf{d}(\mathcal{T}) = 0$. In other words, $\mathbf{d}(\mathcal{T})$ is the maximal distance from the root to another point of the tree.
\end{definition}

Let $A \subset \mathbb{N}$ with $\# A = n \in \mathbb{N}$. Given a rooted planar tree $[\mathcal{T}]$, we can label it from the root to its predecessors and from the right to the left (see Figure \ref{f:figure_1}). Let us assume that $\mathcal{T} = (A,\prec)$ is a tree structure representing the rooted planar tree $[\mathcal{T}]$. For any $0 \leq l \leq \mathbf{d}(\mathcal{T})$, let $(a_1, \ldots, a_{i_l})$ be the subset of points of $A$ that are at distance $l$ from the root, ordered so that $a_1 < \cdots < a_{i_l}$. Set
$\upsilon_l(\mathcal{T}) := (\upsilon^l_1, \ldots, \upsilon^l_{i_l}) \in \mathbb{N}^{i_l}$ where 
$$
\upsilon^l_j = \# \{ \text{immediate predecessors of }a_j\text{ in }\mathcal{T} \}, \quad j = 1, \ldots, i_l.
$$ 
Here, we have $i_0 = 1$ and:
$$
i_l = \sum_{j=1}^{i_{l-1}} \upsilon^{l-1}_j, \quad 1 \leq l \leq \mathbf{d}(\mathcal{T}).
$$ 
With $\mathcal{T} = (A, \prec)$ we associate the vector:
\begin{equation}
\label{e:labelling}
\delta(\mathcal{T}) = \upsilon_{\mathbf{d}}(\mathcal{T}) \times \cdots \times \upsilon_0(\mathcal{T}),
\end{equation}
where $\mathbf{d} = \mathbf{d}(\mathcal{T})$, and the product of two vectors is given by
$$
(a_1, \ldots, a_{l_1}) \times (b_1, \ldots b_{l_2}) = (a_1, \ldots, a_{l_1}, b_1, \ldots, b_{l_2}).
$$
In particular, $\upsilon^0_1 = \delta_n = i_1$. Notice that if $\mathcal{T} \sim \mathcal{T}'$ then $\delta(\mathcal{T}) = \delta(\mathcal{T}')$.

\begin{lemma}
\label{l:trees_and_oreders}
For any simple index set $(A,\delta)$ there is a unique rooted planar tree $[\mathcal{T}]$ such that \eqref{e:labelling} holds. Conversely, let $\mathcal{T} = (A,\prec)$ be a tree structure, then $\delta$ given by \eqref{e:labelling} defines a simple index set.
\end{lemma}

A proof of Lemma \ref{l:trees_and_oreders} is included in Appendix \ref{a:trees_and_index}. In the sequel, we denote $\upsilon_l(\delta) := \upsilon_l(\mathcal{T})$ and $\mathbf{d}(\delta ) := \mathbf{d}(\mathcal{T})$ where $[\mathcal{T}]$ is the rooted planar tree associated with $\delta$. Notice that if $\delta : A \to \mathbb{N}$ belongs to $\Delta(n)$, the identification $\delta \equiv (\delta_a)_{a\in A}$ allows us to modify the set $A$ without changing the tree defined by $\delta$, just by relabelling its nodes. We will only consider labellings $A$ ordered from the top to the botton and from the right to left, as described by \eqref{e:labelling} (see Figure \ref{f:figure_1}).

Let $(A,\delta)$ be a simple index set and let $B \subset A$ have a unique maximal element. Then the induced ordering on $B$ defines a rooted planar tree, hence corresponds to a $\delta_B \in \Delta(l)$, $l = \# B$. This $\delta_B$ will be denoted by $\delta/B$. Notice that, in general, $\delta/ B$ is not the restriction of $\delta$ (considered as a mapping $\delta: A \to \mathbb{N}$) to the subset $B$ (see Figure \ref{f:tree_with_examples}). Moreover, we define $\upsilon_l(\delta/B)$ with respect to the distance $l$ from the root of $B$ (and not from the root of $A$).  Let $(\delta,A)$ be a simple index set and $a \in A$. We call $A \setminus \{a\} = A_1 \cup \cdots \cup A_r$ the \textbf{natural decomposition of $A \setminus \{a\}$ into simple index sets} if $A_\iota = A(a_\iota)$ with $a_\iota$ an immediate predecessor of $a$ (for the tree $\mathcal{T}$ given by \eqref{e:labelling}) for each $\iota \in \{1,\ldots,r \}$ and if $\iota < \jmath$, then $a_\iota > a_\jmath$. 

\begin{definition}
Let $\delta^j \in \Delta(n_j)$ with $n_j \in \mathbb{N}$ for $j = 1,2$. Let $\delta^j : A^j \to \mathbb{Z}^d$, and let $a \in A^2$ be at distance $l(a)$ of the root, and $A(a) \setminus \{a \} = A^2_1 \cup \cdots \cup A^2_r$ be the natural decomposition of $A(a) \setminus \{a \} \subset A^2$ into simple index sets. We define the tree $\delta = \delta^1 \triangleleft_\iota^a \delta^2$ (see Figure \ref{f:connections}) as the tree resulting of connecting $\delta^1$ with $\delta^2$ through the node $a$ of $\delta^2$ at position $0 \leq \iota \leq r$. That is, assuming that $A^1 \cap A^2 = \emptyset$ (changing $A^2$ if necessary), $\delta : A^1 \cup A^2 \to \mathbb{N}$ is the simple index set satisfying $\delta/ A^j = \delta^j$,
and
$$
\upsilon_l(\delta/A(a)\setminus \{a\}) = \upsilon_{l}(\delta^2/A^2_1) \times \cdots \times \upsilon_{l}(\delta^2/A^2_\iota) \times \upsilon_{l}(\delta^1) \times \upsilon_{l}(\delta^2/A^2_{\iota+1}) \times \cdots \times \upsilon_{l}(\delta^2/A^2_r),
$$ 
for $0 \leq l \leq \max \{ \mathbf{d}(\delta^1), \mathbf{d}(\delta^2/A(a)\setminus \{a \}) \}$, where $\upsilon_l(\delta') = \emptyset$ if $l \notin \{0, \ldots, \mathbf{d}(\delta') \}$ for any $\delta' \in \Delta(n')$.

More generally, let $\delta^i \in \Delta(n_i)$ such that $\delta^i : A^i \to \mathbb{N}$ is a simple index set for $i \in \{1, \ldots, j \}$. We define the simple index set $\delta : A^1 \cup \cdots \cup A^j \to \mathbb{N}$ (changing the $A^i$ if necessary to respect the labelling \eqref{e:labelling}) resulting of connecting $\delta^i$ with $\delta^j$ (for $i \in \{ 1, \ldots, j -1 \}$) through the node $a_i \in A^j$ at position $\iota_i$. That is: 
$$
\delta = \delta^1 \triangleleft_{\iota_1}^{a_1} \delta^2 \triangleleft_{\iota_2}^{a_2} \cdots \triangleleft_{\iota_{j-2}}^{a_{j-2}} \delta^{j-1} \triangleleft_{\iota_{j-1}}^{a_{j-1}} \delta^j :=   \delta^1 \triangleleft_{\iota_1}^{a_1} ( \delta^2 \triangleleft_{\iota_2}^{a_2} \cdots \triangleleft_{\iota_{j-2}}^{a_{j-2}} ( \delta^{j-1} \triangleleft_{\iota_{j-1}}^{a_{j-1}} \delta^j) \cdots).
$$

\end{definition}

\begin{center}
\begin{figure}[h]
\includegraphics[scale=0.5]{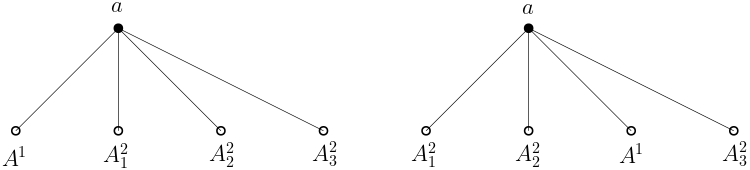}
\caption{Let $\delta^1 : A^1 \to \mathbb{Z}^d$ and $\delta^2 : A^2 \to \mathbb{Z}^d$ two simple index sets. Let $A(a) \setminus \{a\} = A^2_1 \cup A^2_2 \cup A^2_3$. On the left, we represent $\delta^1 \triangleleft_0^a \delta^2$. On the right, we represent $\delta^1 \triangleleft_2^a \delta^2$. We assume without loss of generality that $A^1 \cap A^2 = \emptyset$ and that the tree induced by $(\delta, A_1 \cup A_2)$ is labelled according to \eqref{e:labelling}. }
\label{f:connections}
\end{figure}
\end{center}


\subsection{Semiclassical pseudodifferential calculus} 

To describe properly the solution to \eqref{e:cohomological_general_equation}, we first need to introduce some semiclassical calculus.  We start by recalling the following formula for the commutator of two symbols using the notations of Appendix \ref{s:tools_of_analytic_calculus}. 
Let $H, a \in \mathcal{A}_s(T^* \mathbb{T}^d)$, by \eqref{e:symbol_commutator} we have
\begin{equation}
\label{e:def_symbol_commutator}
[H,a]_\hbar(z) = \frac{2}{\hbar} \int_{\mathcal{Z}^{d} \times \mathcal{Z}^{d}} \mathcal{F}H(w_1) \mathcal{F}a(w) \sin \left( \frac{\hbar}{2} \{ w_1, w \} \right) \frac{e^{i(w+w_1)\cdot z}}{(2\pi)^{2d}} \kappa(dw_1) \, \kappa(dw),
\end{equation}
where $\{w,w' \} = \eta \cdot k' - k \cdot \eta$ is the standard symplectic product, we denote $w = (k,\eta)$ and $w' = (k',\eta')$. We set, for $w,w_1 \in \mathcal{Z}^d := \mathbb{Z}^d \times \R^d$:
\begin{align}
\label{e:sigma_definition}
\sigma^1_\hbar(w,w_1) & := \frac{2}{\hbar} \sin \left( \frac{\hbar}{2} \{ w , w_1 \} \right).
\end{align}
We next generalize this expression and write succesive commutators of symbols, as those appearing in \eqref{e:general_coefficient}, in terms of Fourier multipliers using the space $\mathcal{Z}^{d}$. We define by recursive formula, for $j \geq 2$, $w, w_1, \ldots, w_j \in \mathcal{Z}^{d}$:
\begin{align}
\label{e:multilinear_map}
\sigma_\hbar^j(w,w_1,\ldots, w_j) & := \sigma_\hbar^{j-1}(w,w_1,\ldots, w_{j-1}) \sigma_\hbar^1(w+w_1+ \cdots + w_{j-1},w_j).
\end{align}
\begin{lemma}
\label{l:sequences_of_commutators}
The following holds:
\begin{align*}
[H_j, \ldots, [H_1,a]_\hbar \cdots ]_\hbar &  \\[0.2cm]
 & \hspace*{-2cm} = \int_{(\mathcal{Z}^{d})^{j+1}} \mathcal{F}a(w) \mathcal{F}H_1(w_1) \cdots \mathcal{F}H_j(w_j) \sigma_\hbar^j(\mathbf{w})e^{i(w + w_1 + \cdots + w_j) \cdot z} \kappa(d \mathbf{w}),
\end{align*}
where $\mathbf{w} = (w,w_1,\ldots,w_j)$, and $\kappa(\mathbf{w}) = \kappa(w) \kappa(w_1) \cdots \kappa(w_j)$.
\end{lemma}
\begin{proof}
The case $j=1$ has been already shown. To show the general case, we use repeteadly formulas \eqref{e:def_symbol_commutator} and \eqref{e:multilinear_map}. 
\end{proof}



Finally, we introduce the following generalization of the above Fourier multipliers in terms of the diagrammatic trees $\delta \in \Delta(n)$:
 \begin{definition}
Let $n \geq 2$, $\mathbf{w} = (w, w_1, \ldots, w_{n-1}) \in (\mathcal{Z}^{d})^{n}$ and $\delta \in \Delta(n)$. We define:
 \begin{align}
 \label{e:graph_commutator}
 \sigma_\hbar(\mathbf{w}, \delta)  = \sigma_\hbar^{r}(w, \Sigma_1(\mathbf{w}), \ldots, \Sigma_r(\mathbf{w})) \sigma_\hbar(\mathbf{w}/A_1,\delta/A_1) \cdots   \sigma_\hbar(\mathbf{w}/A_r,\delta/A_r),
 \end{align}
 where $\Sigma_l(\mathbf{w}) = \sum_{j \in A_l} w_j$ for $1 \leq  l \leq r$, $A \setminus \{ n \} = A_1 \cup \cdots \cup A_r$ is the natural decomposition of $A \setminus \{n \}$ into simple index sets, and we set $\sigma_\hbar(w,(0)) = 1$.
\end{definition}

\subsection{The Lindstedt series}

We are now in position to describe the tree structure giving the Lindstedt series for the problem \eqref{e:conjugation_equation}. The main result of this section is Theorem \ref{t:main_lindstedt_teor}.

\begin{definition}
We define coefficients $\mathbf{c}((0)) = 1$ and, for any $\delta \in \Delta(n)$ with $n \geq 2$:
$$
\mathbf{c}(\delta) := \mathbf{c}_{k_1, \ldots, k_r} \mathbf{c}(\delta/A_1) \cdots \mathbf{c}(\delta/A_r),
$$
where $k_l = \# A_l$ for $1 \leq l \leq r$ and $A \setminus \{n \}  = A_1 \cup \cdots \cup A_r$ is the natural decomponsition of $A \setminus \{ n \}$ into simple index sets.
\end{definition}

By solving recursively the cohomological equation \eqref{e:cohomological_general_equation}, we obtain the Lindstedt series for the formal solutions $H(t)$ and $R'(t)$ of \eqref{e:conjugation_equation}:
\begin{teor}
\label{t:main_lindstedt_teor} For every $n \geq 1$, the solution to the cohomological equation \eqref{e:cohomological_general_equation} is given by:
\begin{align}
H_{n}(x,\xi) & = \sum_{ \substack{ v \in (\mathbb{Z}^d)^n \\[0.1cm] \delta \in \Delta(n)}}  \int_{(\R^d)^n} \Omega_1(\delta,v) \widehat{\mathcal{F}}_\hbar(\delta,v,\eta) e^{i(v_1 + \cdots + v_n) \cdot x} e^{i (\eta_1 + \cdots + \eta_n) \cdot \xi} d\eta, \\[0.2cm]
R'_n(\xi) & = \sum_{ \substack{ v \in (\mathbb{Z}^d)^n \\[0.1cm] \delta \in \Delta(n)}} \int_{(\R^d)^n} \Omega_2(\delta,v) \widehat{\mathcal{F}}_\hbar(\delta,v,\eta) e^{i(\eta_1 + \cdots + \eta_n) \cdot \xi} d\eta,
\end{align}
where $v = (v_1, \ldots, v_n) \in (\mathbb{Z}^d)^n$, $\eta = (\eta_1, \ldots , \eta_n) \in (\mathbb{R}^d)^n$, and
\begin{equation}
\label{e:tree_derivative}
\widehat{\mathcal{F}}_\hbar(\delta,v,\eta) = \mathbf{c}(\delta)  \widehat{V}(w_1) \cdots \widehat{V}(w_n) \sigma_\hbar(\mathbf{w},\delta).
\end{equation}
Moreover, the coefficients $\Omega_1(\delta,v)$ are given by recursive formula described as follows. For any $v \in (\mathbb{Z}^d)^n$, denote $\Sigma(v) := v_1 + \cdots + v_n$; then, for $n = 1$, we have: $\Omega(\delta,v) = 0$, 
$$
\Omega_1(\delta,v) := \left \lbrace \begin{array}{ll} 0, & \text{if } v = 0; \\[0.2cm]
\displaystyle \frac{1}{i v \cdot \omega}, & \text{if } v \neq 0,
\end{array} \right. \quad \quad \Omega_2(\delta,v) := \left \lbrace \begin{array}{ll} 0, & \text{if } v \neq 0; \\[0.2cm]
\displaystyle 1, & \text{if } v = 0.
\end{array} \right.
$$ 
While for $n \geq 2$, the coefficients $\Omega_1$ and $\Omega_2$ satisfy:
\begin{align}
\label{e:Omega_1}
\Omega_1(\delta,v)  & = \left \lbrace \begin{array}{ll} \displaystyle \frac{1}{i \Sigma(v) \cdot \omega}    \big( \Omega_1(\delta,v/A_1) \cdots \Omega_1(\delta,v/A_r) - \Omega(\delta,v) \big), &  \Sigma(v) \neq 0, \\[0.4cm]
  0, &  \Sigma(v) = 0;
 \end{array} \right.
\end{align}
where $A \setminus \{ n \} = A_1 \cup \cdots \cup A_r$ is the natural decomposition into simple index sets, 
\begin{equation}
\label{e:definition_omega}
\Omega(\delta,v) = \sum_{B \in \Gamma_1(\delta)} \Omega_2(\delta,v/A\setminus A(B)) \Omega_1(\delta,v/A(b_1)) \cdots \Omega_1(\delta,v/A(b_s));
\end{equation}
where $\Gamma_1(\delta)$ is the family of subsets $B \subset A \setminus \{ n \}$ such that $B = \{b_1, \ldots, b_s \}$ consists of pairwise unrelated elements in $A \setminus \{ n \}$, and
\begin{equation}
\label{e:Omega_2}
\Omega_2(\delta,v) = \left \lbrace \begin{array}{ll} 0,  &  \Sigma(v) \neq 0, \\[0.4cm]\displaystyle    \Omega_1(\delta,v/A_1) \cdots \Omega_1(\delta,v/A_r) - \Omega(\delta,v,a), &  \Sigma(v) = 0.
 \end{array} \right.
\end{equation}
\end{teor}

\begin{remark}
The coefficients $\Omega_1(\delta,v)$ and $\Omega_2(\delta,v)$ coincide with those given in \cite[Lemma 3]{Elia89}. 
\end{remark}

\begin{center}
\begin{figure}[h]
\includegraphics[scale=0.5]{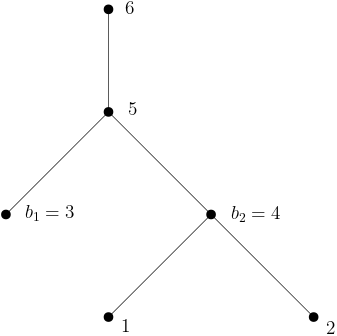}
\caption{Tree diagram with $n = 6$, $\delta = (0,0,0,2,2,1 )$, $n= 6$, $B = \{b_1,b_2 \} = \{3,4 \}$, $A \setminus A(B) = \{ 5,6 \}$, $\delta/ A(b_1) = ( 0)$, $\delta/ A(b_2) = (0,0,2 )$, and $\delta/ A \setminus A(B) = (0,1 )$.}
\label{f:tree_with_examples}
\end{figure}
\end{center} 

\begin{proof}
We proceed by induction. The case $n = 1$ follows by equation \eqref{e:first_cohomological_equation}. For the general case we invoke \eqref{e:cohomological_general_equation}. Let us split 
$$
H_n = \sum_{\delta \in \Delta(n)} H_\delta, \quad R'_n = \sum_{\delta \in \Delta(n)} R'_\delta,
$$
and assume that $\widehat{H}_\delta(v,\eta) = \Omega_1(\delta,v) \widehat{\mathcal{F}}_\hbar(\delta,v,\eta)$, for $\delta \in \Delta(k)$ with $1 \leq k \leq n-1$.  Then, by \eqref{e:cohomological_general_equation},
\begin{align*}
i \omega \cdot \nabla_x H_{n} & = \sum_{r=1}^{n-1} \sum_{k_1 + \cdots + k_r = n-1} \sum_{\delta^\jmath \in \Delta(k_\jmath)} \mathbf{c}_{k_1,\ldots, k_r} [H_{\delta^r}, \ldots, [ H_{\delta^1}, V - R'_1]_\hbar \cdots ]_\hbar  \\[0.2cm]
 & \quad - \sum_{m = 2}^{n-1} \sum_{s=1}^{n-m} \sum_{l_1 + \cdots + l_s = n-m} \sum_{\delta^\iota \in \Delta(l_\iota), \; \delta^0 \in \Delta(m)} \mathbf{c}_{l_1, \ldots, l_s} [H_{\delta^s}, \ldots, [H_{\delta^1}, R'_{\delta^0}]_\hbar \cdots]_\hbar - R'_n.
\end{align*}
From the first term of the right-hand-side, using that $\mathbf{c}(\delta) = \mathbf{c}_{k_1,\ldots, k_r} \mathbf{c}(\delta^1) \cdots \mathbf{c}(\delta^r)$, where $A \setminus \{n \} =  A_1 \cup \cdots \cup A_r$ is the natural decomposition into simple index sets, denoting $\delta(n) = r$ and $\delta/A_l = \delta^l$ for $1 \leq l \leq r$, we obtain the contribution
$$
 \Omega_1(\delta,v/A_1) \cdots \Omega_1(\delta,v/A_r) \widehat{\mathcal{F}}_\hbar(\delta,v,\eta)
$$
in expression \eqref{e:Omega_1} once we solve the cohomological equation using \eqref{e:solution_cohomological_equation}. On the other hand, to study the contribution of the second term, we recall definition \eqref{e:coefficients} of the coefficients $\mathbf{c}_{l_1,\ldots,l_s}$ and consider first the commutator
\begin{equation}
\label{e:first_considered_commutator}
-\frac{1}{l_1} \mathbf{c}(\delta^1) \mathbf{c}(\delta^0) [H_{\delta^1}, R'_{\delta^0}]_\hbar.
\end{equation}
We aim at showing by using repeteadly the Jacobi rule \eqref{e:jacobi_in_fourier} that this commutator generates a sum of terms connecting the tree $\delta^1$ with $\delta^0$ through the nodes of $\delta^0$. Precisely, if $\delta^0 : A_0 \to \mathbb{N}$ decomposes into simple index sets as $A_0 \setminus \{ n \} =  A^0_1 \cup \cdots \cup A^0_i$ with $\delta/A^0_\iota = \delta^0_\iota$, then denoting $l^0_\iota = \# A^0_\iota$, we have
$$
\frac{1}{l_1} \mathbf{c}(\delta^1) \mathbf{c}(\delta^0) = \frac{1}{l_1} \mathbf{c}_{l^0_1, \ldots, l^0_i} \mathbf{c}(\delta^0_1) \cdots \mathbf{c}(\delta^0_i).
$$ 
Moreover, by the induction hypothesis, we have that
\begin{align}
\label{e:contribution_term}
[H_{\delta^1}, R'_{\delta^0}]_\hbar & \\[0.2cm]
 & \hspace*{- 1.5cm} = \int \sigma_\hbar^1(\Sigma(\mathbf{w}^1), \Sigma(\mathbf{w}^0))  \Omega_2(\delta^0,v^0) \Omega_1(\delta^1,v^1) \widehat{\mathcal{F}}_\hbar(\delta^1,\mathbf{w}^1) \widehat{\mathcal{F}}_\hbar(\delta^0, \mathbf{w}^0) e^{i ( \Sigma(\mathbf{w}^1) + \Sigma(\mathbf{w}^0)) \cdot z} \kappa(d\mathbf{w}^1, d\mathbf{w}^0). \notag
\end{align}
Then we observe that:
\begin{align*}
\sigma_\hbar^1(\Sigma(\mathbf{w}^1), \Sigma(\mathbf{w}^0)) \sigma_\hbar(\mathbf{w}^1,\delta^1) \sigma_\hbar(\mathbf{w}^0,\delta^0) =  \sigma_\hbar( \mathbf{w}^1, \mathbf{w}^0, \delta^1 \triangleleft_0^n \delta^0).
\end{align*}
The term \eqref{e:contribution_term} contributes to $- \Omega_2(\delta^0,v^0) \Omega_1(\delta^1,v^1)  \widehat{\mathcal{F}}_\hbar(\delta,\mathbf{w})$ when $A = A_1 \cup A_0$ and $\delta = \delta^1 \triangleleft_0^n \delta^0$. However, the coefficient $\frac{1}{l_1} \mathbf{c}(\delta^1) \mathbf{c}(\delta^0)$ appearing in front of this term is in general larger than $\mathbf{c}(\delta^1 \triangleleft_0^n \delta^0)$. This means that this coefficient has to be splitted in a way that the commutator \eqref{e:first_considered_commutator} produces all terms corresponding with $- \Omega(\delta,v)\widehat{\mathcal{F}}_\hbar(\delta,\mathbf{w})$ when the tree $\delta$ is obtained connecting $\delta^1$ with $\delta^0$ through any node of $\delta^0$ (not necessarily the root) and in any other position on the plane.  Denoting $w = \mathbf{w}^0/ \{ n \}$, we have:
\begin{align*}
\sigma_\hbar( \mathbf{w}^1, \mathbf{w}^0, \delta^1 \triangleleft_0^n \delta^0) & \\[0.2cm]
 & \hspace*{-2.5cm} =  \sigma_\hbar^1(\Sigma(\mathbf{w}^1), \Sigma(\mathbf{w}^0)) \sigma_\hbar^i(w,\Sigma_1(\mathbf{w}^0), \ldots , \Sigma_i(\mathbf{w}^0)) \sigma_\hbar( \mathbf{w}^1,\delta^1) \prod_{\iota = 1}^i \sigma_\hbar(\mathbf{w}^0/A^0_\iota,\delta^0_\iota) \\[0.2cm]
 & \hspace*{-2.5cm} = \sigma_\hbar^1(\Sigma(\mathbf{w}^1), \Sigma(\mathbf{w}^0)) \sigma_\hbar^1(\Sigma_i(\mathbf{w}^0), w + \Sigma_1(\mathbf{w}^0) + \cdots + \Sigma_{i-1}(\mathbf{w}^0)) \\[0.2cm]
 & \hspace*{-2.5cm} \quad \quad \quad \quad \quad \times \sigma_\hbar^{i-1}(w,\Sigma_1(\mathbf{w}^0), \ldots , \Sigma_{i-1}(\mathbf{w}^0)) \sigma_\hbar( \mathbf{w}^1,\delta^1) \prod_{\iota = 1}^i \sigma_\hbar(\mathbf{w}^0/A^0_\iota,\delta^0_\iota).
\end{align*}
Now, by using the Jacobi rule \eqref{e:jacobi_in_fourier}, we have
\begin{align*}
\sigma_\hbar^1(\Sigma(\mathbf{w}^1), \Sigma(\mathbf{w}^0)) \sigma_\hbar^1(\Sigma_s(\mathbf{w}^0), w + \Sigma_1(\mathbf{w}^0) + \cdots + \Sigma_{i-1}(\mathbf{w}^0)) & \\[0.2cm]
 & \hspace*{-9cm} = \sigma_\hbar^1 ( \Sigma(\mathbf{w}^1) + \Sigma_i(\mathbf{w}^0), w + \Sigma_1(\mathbf{w}^0) + \cdots + \Sigma_{i-1}(\mathbf{w}^0)) \sigma_\hbar^1 ( \Sigma(\mathbf{w}^1), \Sigma_i(\mathbf{w}^0))  \\[0.2cm]
 & \hspace*{-9cm} \quad + \sigma_\hbar^1 ( \Sigma_i(\mathbf{w}^0), \Sigma(\mathbf{w}^1)+ w + \Sigma_1(\mathbf{w}^0) + \cdots + \Sigma_{i-1}(\mathbf{w}^0)) \\[0.2cm]
 & \hspace*{-9cm} \hspace*{3cm} \times \sigma_\hbar^1(\Sigma(\mathbf{w}^1), w + \Sigma_1(\mathbf{w}^0) + \cdots + \Sigma_{i-1}(\mathbf{w}^0)).
\end{align*}
Therefore, denoting $n_i \in A_0$ such that $A^0_i = A(n_i)$, we get:
$$
\sigma_\hbar( \mathbf{w}^1, \mathbf{w}^0, \delta^1 \triangleleft_0^n \delta^0) = \sigma_\hbar(\mathbf{w}^1,\mathbf{w}^0, \delta^1 \triangleleft_0^{n_i} \delta^0) + \sigma_\hbar(\mathbf{w}^1,\mathbf{w}^0, \delta^1 \triangleleft_1^n \delta^0).
$$
This procedure can be iterated, expanding next these two terms in the above expression again by the Jacobi rule, reducing the identification of the coefficients to an inductive scheme covering all the nodes of $\delta^0$ and all possible positions in the rooted planar tree obtained by connecting $\delta^1$ with $\delta^0$.  It remains to show that all these terms given by this procedure can be weighted to appear with the corresponding coefficients $\mathbf{c}(\delta)$ from the splitting of the coefficient $\frac{1}{l_1} \mathbf{c}(\delta^1) \mathbf{c}(\delta^0)$. To see this, we use Lemma \ref{l:jacobi_coefficients}, which gives us: 
\begin{align*}
\frac{1}{l_1} \mathbf{c}(\delta^1) \mathbf{c}(\delta^0) &  =  \mathbf{c}(\delta^1) \mathbf{c}(\delta^0_1) \cdots \mathbf{c}(\delta^0_i) \frac{1}{l_1} \mathbf{c}_{l^0_1, \ldots, l^0_i} \\[0.2cm]
 & = \mathbf{c}(\delta^1) \mathbf{c}(\delta^0_1) \cdots \mathbf{c}(\delta^0_i) \left( \mathbf{c}_{l^0_1, \ldots, l^0_i,l_1} + \frac{1}{l_1} \mathbf{c}_{l^0_1, \ldots, l^0_i + l_1}  \right).
\end{align*}

This splitting matches with the iteration of the Jacobi rule, so that:
\begin{align*}
\frac{1}{l_1} \mathbf{c}(\delta^1) \mathbf{c}(\delta^0) \sigma_\hbar( \mathbf{w}^1, \mathbf{w}^0, \delta^1 \triangleleft_0^n \delta^0) & \\[0.2cm]
 & \hspace*{-3cm} = \mathbf{c}(\delta^1 \triangleleft_0^n \delta^0) \sigma_\hbar( \mathbf{w}^1, \mathbf{w}^0, \delta^1 \triangleleft_0^n \delta^0) \\[0.2cm]
 & \hspace*{-3cm}   \quad +  \frac{1}{l_1}  \mathbf{c}_{l^0_1, \ldots, l^0_i + l_1} \mathbf{c}(\delta^1) \mathbf{c}(\delta^0_1) \cdots \mathbf{c}(\delta^0_i) \sigma_\hbar(\mathbf{w}^1,\mathbf{w}^0, \delta^1 \triangleleft_0^{n_i} \delta^0) \\[0.2cm]
 & \hspace*{-3cm}  \quad +  \frac{1}{l_1} \mathbf{c}_{l_1^0, \ldots, l_{i}^0+l_1} \mathbf{c}(\delta^1) \mathbf{c}(\delta^0_1) \cdots \mathbf{c}(\delta^0_i)  \sigma_\hbar(\mathbf{w}^1,\mathbf{w}^0, \delta^1 \triangleleft_1^n \delta^0).
\end{align*}
Iterating this procedure by induction to cover all the nodes and positions in $\delta^0$, we obtain:
\begin{align*}
\frac{1}{l_1} \mathbf{c}(\delta^1) \mathbf{c}(\delta^0) \sigma_\hbar( \mathbf{w}^1, \mathbf{w}^0, \delta^1 \triangleleft_0^n \delta^0) & = \sum_{a \in A_0} \sum_{\iota \in \{0, \ldots, \delta^0_a\}} \mathbf{c}(\delta^1 \triangleleft_\iota^a \delta^0) \sigma_\hbar( \mathbf{w}^1, \mathbf{w}^0, \delta^1 \triangleleft_\iota^a \delta^0) \\[0.2cm]
 & = \sum_{\delta \in \Delta(\delta^0, \delta^1)} \mathbf{c}(\delta) \sigma_\hbar( \mathbf{w}^1, \mathbf{w}^0, \delta),
\end{align*}
where $\Delta(\delta^0, \delta^1)$ is the set of trees $\delta : A \to \mathbb{Z}$ such that:
$$
A \setminus A_0 = A(b_1) = A_1,
$$
where $b_1 \in A \setminus A_0$ satisfies $b_1 \prec n$, and so that $\delta/A_\iota = \delta^\iota$ for $\iota = 0 ,1$. On the other hand, given $(l_1,\ldots, l_s)$ such that $l_1 + \cdots + l_s = n -m$, and let $\delta^\iota \in \Delta(l_\iota)$ for $\iota = 1, \ldots, s$, and $\delta^0 \in \Delta(m)$, the contribution of the commutator 
$$
-\sum_{\sigma \in \pi_s} \mathbf{c}_{l_{\sigma(1)}, \ldots, l_{\sigma(s)}} [H_{\delta^{\sigma(s)}}, \ldots, [H_{\delta^{\sigma(1)}}, R'_{\delta^0}]_\hbar \cdots]_\hbar
$$
can be treated by similar arguments.  We use Jacobi rule \eqref{e:jacobi_in_fourier} and  Lemmas \ref{l:ordered_permutations} and \ref{l:transposition_lemma} to generalize the above iterative argument. 
Given $1 \leq r_1 \leq r_2 \leq s$, let us consider the partition $\{1, \ldots, s \} =  I^{r_1,r_2}_1 \cup I^{r_1,r_2}_2 \cup I^{r_1,r_2}_3$  given by
\begin{align*}
I^{r_1,r_2}_1 & = \{ 1, \ldots, r_1 \}, \\[0.2cm]
I^{r_1,r_2}_2 & = \{ r_1 + 1, \ldots,  r_2 \}, \\[0.2cm]
I^{r_1,r_2}_3 & = \{r_2+1, \ldots,  s \}. 
\end{align*} 
In particular, if $r_1 = r_2$ we set $I^{r_1,r_2}_2 = \emptyset$, and if $r_2 = s$ we set $I^{r_1,r_2}_3 = \emptyset$. Let $\mathcal{C}(s,r_2)$ be the subset of permutations $\varsigma \in \pi_s$ of the form
$$ 
\varsigma(1, \ldots, s ) = (\jmath_1, \ldots, \jmath_{r_2}, \iota_{r_2+1}, \ldots, \iota_s ),
$$
where $(\iota_{r_2+1}, \ldots, \iota_s)$ is a combination of $\{1,\ldots, s \}$ ordered with respect to the natural order $<$, and $(\jmath_1, \ldots, \jmath_{r_2} )$ is its ordered complement\footnote{In particular $\# \mathcal{C}(s,r_2) = {s \choose r_2}$ since $\mathcal{C}(s,r_2)$ is in one-to-one correspondence with combinations of $r_2$ elements.}.  We define the subgroup   $\Pi(s,r_1,r_2)$ of permutations $\sigma' = (\sigma'_1,\sigma'_2,\sigma'_3) \in \pi_s$ such that $\sigma'(I^{r_1,r_2}_j) = I^{r_1,r_2}_j$ for $j=1,2,3$. Notice that for every $\sigma \in \pi_s$, there exist unique $\varsigma \in \mathcal{C}(s,r_2)$, $\sigma_* \in \mathcal{O}(r_1,r_2)$, and $\sigma' \in \Pi(s,r_1,r_2)$ such that
\begin{equation}
\label{e:decomposition_permutation}
\sigma =  \big( \sigma_*, \operatorname{Id}_3 \big) \circ \sigma' \circ \varsigma,
\end{equation}
where $\operatorname{Id}_3$ is the identity permutation on $(I^{r_1,r_2}_3,<)$. For any $\sigma \in \pi_s$, we set:
$$
\delta_{\sigma} :=  \delta^{\sigma(s)} \triangleleft_0^n \cdots \delta^{\sigma(r_2+1)} \triangleleft_0^n \delta^{\sigma(r_2)} \triangleleft_{0}^{n_i} \cdots \triangleleft_{0}^{n_i} \delta^{\sigma(r_1+1)} \triangleleft_0^{n_i} \delta^{\sigma(r_1)} \triangleleft_1^n \cdots \triangleleft_1^n \delta^{\sigma(1)} \triangleleft_1^n \delta^0,
$$ 
where $n$ is the root of $\delta^0 = \delta_\sigma/A_0$ for the labelling \eqref{e:labelling} of the tree induced by $(\delta_\sigma, A)$ with $A = A_0 \cup A_1 \cup \cdots \cup A_s$, and $n_i$ is the root of $\delta^0_i = \delta_\sigma/ A^0_i$, where $A_0 \setminus \{n \} = A^0_1 \cup \cdots \cup A^0_i$ is the natural decomposition into simple index sets. 
Using the Fourier inversion formula (similarly as in \eqref{e:contribution_term}) for the commutator
$$
- \sum_{\sigma \in \pi_s} \mathbf{c}_{l_{\sigma(1)}, \ldots, l_{\sigma(s)}} [H_{\delta^{\sigma(s)}}, \ldots, [H_{\delta^{\sigma(1)}}, R'_{\delta^0}]_\hbar \cdots]_\hbar,
$$
we lead to study the Fourier multiplier
$$
-\sum_{\sigma \in \pi_s} \mathbf{c}_{l_{\sigma(1)}, \ldots, l_{\sigma(s)}} \mathbf{c}(\delta^s) \cdots \mathbf{c}(\delta^1) \mathbf{c}(\delta^0) \sigma_\hbar(\mathbf{w}, \delta^{\sigma(s)} \triangleleft_0^n \cdots \triangleleft_0^n \delta^{\sigma(1)} \triangleleft_0^n \delta^0).
$$
By Jacobi rule \eqref{e:jacobi_in_fourier} and Lemma \ref{l:transposition_lemma}, we obtain:
\begin{align*}
\sum_{\sigma \in \pi_s} \mathbf{c}_{l_{\sigma(1)}, \ldots, l_{\sigma(s)}} \mathbf{c}(\delta^s) \cdots \mathbf{c}(\delta^1) \mathbf{c}(\delta^0) \sigma_\hbar(\mathbf{w}, \delta^{\sigma(s)} \triangleleft_0^n \cdots \triangleleft_0^n \delta^{\sigma(1)} \triangleleft_0^n \delta^0) & \\[0.2cm]
  &  \hspace*{-9cm} =  \sum_{r_2= 1}^{s} \sum_{r_1=1}^{r_2} \sum_{\sigma \in \pi_s} \mathbf{C}_{s,r_1,r_2}(\sigma) \mathbf{c}(\delta^{s}) \cdots \mathbf{c}(\delta^{1}) \mathbf{c}(\delta^0_1) \cdots  \mathbf{c}(\delta^0_{i}) \sigma_\hbar(\mathbf{w},\delta_\sigma),
\end{align*}
where 
$$
\mathbf{C}_{s,r_1,r_2}(\sigma) =  \mathbf{c}_{l_1^0, \ldots, l^0_i + l_{\sigma(1)} + \cdots + l_{\sigma(r_2)}, l_{\sigma(r_2+1)}, \ldots, l_{\sigma(s)}} \mathbf{c}_{l_{\sigma(1)}, \ldots, l_{\sigma(r_2)}}.
$$
Using the decomposition \eqref{e:decomposition_permutation} of $\sigma$, we observe that
$$
\mathbf{c}_{l_1^0, \ldots, l^0_i + l_{\sigma(1)} + \cdots + l_{\sigma(r_2)}, l_{\sigma(r_2+1)}, \ldots, l_{\sigma(s)}} = \mathbf{c}_{l_1^0, \ldots, l^0_i + l_{ \sigma' \circ \varsigma(1)} + \cdots + l_{\sigma' \circ \varsigma(r_2)}, l_{\sigma' \circ \varsigma(r_2+1)}, \ldots, l_{\sigma' \circ \varsigma(s)}},
$$
that is, this coefficient is independent of the permutation $\sigma_* \in \mathcal{O}(r_1,r_2)$. Notice also that
$$
\sigma_\hbar(\mathbf{w},\delta_\sigma) = \sigma_\hbar(\mathbf{w},\delta_{\sigma'\circ \rho})
$$ 
is also independent of the permutation $\sigma_* \in \mathcal{O}(r_1,r_2).$ Moreover, by Lemma \ref{l:ordered_permutations}, we have:
$$
\sum_{\sigma_* \in \mathcal{O}(r_1,r_2)} \mathbf{c}_{l_{\sigma_* \circ \sigma'\circ \varsigma(1)}, \ldots, l_{\sigma_* \circ \sigma'\circ \varsigma(r_2)}} = \mathbf{c}_{l_{\sigma'_1 \circ \varsigma(1)}, \ldots, l_{\sigma'_1 \circ \varsigma(r_1)}} \mathbf{c}_{l_{\sigma'_2 \circ \varsigma(r_1+1)}, \ldots, l_{\sigma'_2 \circ \varsigma(r_2)}}.
$$
Since $\sigma'_1$ runs over $\pi_{r_1}$ and $\sigma'_2$ runs over $\pi_{r_2-r_1}$, this splitting allows us to iterate this procedure, covering then the rest of nodes and positions of $\delta^0$ by an inductive process. Therefore we get:
\begin{align*}
\sum_{\sigma \in \pi_s} \mathbf{c}_{\sigma(l_1, \ldots, l_s)} \mathbf{c}(\delta^s) \cdots \mathbf{c}(\delta^1) \mathbf{c}(\delta^0) \sigma_\hbar(\mathbf{w}^s, \ldots, \mathbf{w}^1, \mathbf{w}^0, \delta^{\sigma(s)} \triangleleft_0^n \cdots \triangleleft_0^n \delta^{\sigma(1)} \triangleleft_0^n \delta^0) & \\[0.2cm] 
 & \hspace*{-10cm} =  \sum_{\delta \in \Delta(\delta^0, \delta^1, \ldots, \delta^s)} \mathbf{c}(\delta) \sigma_\hbar( \mathbf{w}^s, \ldots, \mathbf{w}^1, \mathbf{w}^0, \delta),
\end{align*}
where $n$ is the root of $\delta^0 : A_0 \to \mathbb{Z}$ and $\Delta(\delta^0, \delta^1, \ldots, \delta^s)$ is the set of trees $\delta : A \to \mathbb{Z}$ such that
$$
A \setminus A_0 = A(B) = A(b_1) \cup \cdots \cup A(b_s)
$$
is the decomposition into simple index sets with $B = \{b_1,\ldots, b_s\} \in \Gamma_1(\delta)$, $A_\iota = A(b_\iota)$, and $\delta/A_\iota = \delta^\iota$ for $\iota = 0 ,\ldots, s$. This term contributes with 
$$
- \Omega_2(\delta,v/A\setminus A(B)) \Omega_1(\delta,v/A(b_1)) \cdots \Omega_1(\delta,v/A(b_s)) \widehat{\mathcal{F}}_\hbar(\delta,v,\eta).
$$
Finally, solving the cohomological equation \eqref{e:cohomological_general_equation} by using \eqref{e:solution_cohomological_equation} we obtain the expression for $\Omega_1$. Similar considerations give the formula for $\Omega_2$. This concludes the proof.

\end{proof}

\section{Convergence of the Lindstedt series}
\label{s:convergence_of_the_lindstedt_series}

The coefficients $\Omega_1$ and $\Omega_2$ are sums of products of small divisors. The very technical study of these terms is the heart of the works of Eliasson \cite{Elia89,Eliasson90,Eliasson96}. Proving the convergence of the series giving $H_n$ (and $R'_n$) is a real challenge, since this series is \textit{absolutely divergent} (see \cite{Eliasson96}) and it is necessary to exploit very precise cancelations of signs between terms in this series to show its convergence.  Here we reduce the proof to Lemma \ref{l:lemma_eliasson}, which is consequence of \cite[Prop. 2]{Elia89}, used here as black-box, and from this key lemma we adapt and complete some parts of the work \cite{Elia89} to give a more compact exposition. Finally, in Lemma \ref{e:analytical_part} we show the key estimate regarding the analytic semiclassical pseudodifferential calculus.

\subsection{Admissible families of resonances.} We start by generalizing slightly the concept of resonance introduced in \cite[Def. page 20]{Elia89}. Given $v \in (\mathbb{Z}^d)^n$ and $\delta \in \Delta(n)$, let $\gamma = \gamma_{\delta,v} : A \to \mathbb{Z}^d$
be the map: 
$$
\gamma(a) := \sum_{b \in A(a)} v(b).
$$
Let $A' \subset A$ be such that $\delta/A'$ is a simple index set. We define $\gamma/A' := \gamma_{\delta,v/A'}$.
\begin{definition}
A $\gamma$-resonance (we will call it simply a resonance) is a pair $(B,a) \in A^r \times A$ with  $r \geq 0$, $B = (b_1, \ldots, b_r) \subset A(a) \setminus \{a\}$ of pairwise unrelated elements, such that
$$
\gamma(a) = \gamma(b_1) + \cdots + \gamma(b_r).
$$
We denote $\mathcal{B}_R := A(a) \setminus A(B)$, and notice that we can identify $R \equiv \mathcal{B}_R$. We emphasize that the case $r = 0$ is also covered, so in this case $B = \emptyset$ and $R = (\emptyset,a)$.
\end{definition}

If $R = (B,a)$ is a resonance, then
$$
\sum_{b \in \mathcal{B}_R} v(b) = 0.
$$

\begin{figure}[h]
\includegraphics[scale=0.5]{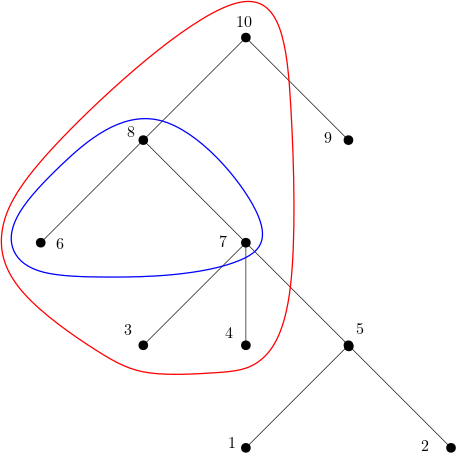}
\caption{Tree with two resonances $\color{blue} R_1 = (B_1,a_1)$ and $\color{red} R_2 = (B_2,a_2)$, where $B_1 = \{3,4,5\}$, $a_1 = 8$, $B_2 = \{ 5,9 \}$, $a_2 = 10$. We encircle the sets $\color{blue} \mathcal{B}_{R_1}$ and $\color{red} \mathcal{B}_{R_2}$.}
\label{f:resonances_1}
\end{figure}

 \begin{definition}
Let $R_1$ and $R_2$ be two resonances.
\begin{enumerate}

\item We say that $R_1 \subset R_2$ if $\mathcal{B}_{R_1} \subset \mathcal{B}_{R_2}$.

\item We say that $R_1$ and $R_2$ are disjoint if $\mathcal{B}_{R_1} \cap \mathcal{B}_{R_2} = \emptyset$.

\end{enumerate}
If (1) or (2) are satisfied, we say that $R_1$ and $R_2$ are non-overlapping.
\end{definition}

Let $J$ be a family of non-overlapping resonances. We define its support by
$$
\operatorname{supp} J := \bigcup_{(B,a) \in J} ]b_1, a[ \; \cup \cdots \cup \;]b_r,a[.  
$$
In the example of Figure \ref{f:resonances_1}, we have $(B_1,a_1) \subset (B_2,a_2)$. If $J = \{ (B_1,a_1), (B_2,a_2) \}$, then $\operatorname{supp} J = \{6,8\}$.

For any $c \in A$, we define $\gamma_J(c)$ in the following way: if $c \notin \operatorname{supp} J$ then $\gamma_J(c) = \gamma(c)$. While if $(B,a)$ is the smallest resonance of $J$ such that $c \in ]b_1, a[ \; \cup \cdots \cup \;]b_r,a[$, then 
$$
\gamma_J(c) :=  \sum_{ b \in  A(c) \setminus A(B) } v(b).
$$
If $\gamma_J(c) \neq 0$ for every $c \in A$ we say that $J$ is \textbf{admissible}. Let $\operatorname{ad}(\gamma)$ the set of all admissible families $J$. We also set $\operatorname{ad}^*(\gamma)$ the set of families of resonances $J$ such that
$$
J \in \operatorname{ad}^*(\gamma) \Leftrightarrow \left \lbrace \begin{array}{ll}
\gamma_J(c) \neq 0, & \forall c \in A \setminus \{ n \}, \\[0.2cm]
 \gamma_J(n) = 0. & \end{array} \right.
$$ 

\begin{lemma}
\label{l:first_description}
The following two expressions for the coefficients $\Omega_1$ and $\Omega_2$ hold:
\begin{align}
\label{e:first_omega}
\Omega_1(\delta,v) & = \sum_{ J \in \operatorname{ad}(\gamma)} \prod_{c \in A} (-1)^{\# J} ( i \omega \cdot \gamma_{J}(c))^{-1}, \\[0.2cm]
 \Omega_2(\delta,v) & = \sum_{ J \in \operatorname{ad}^*(\gamma)} \prod_{c \in A \setminus \{n \}} (-1)^{\# J} ( i \omega \cdot \gamma_{J}(c))^{-1}.
\end{align}
\end{lemma}

\begin{proof}
 We use the recursive definition of the coefficients $\Omega_1(\delta,v)$ given in Theorem \ref{t:main_lindstedt_teor}. Indeed, let $A \setminus \{ n \} = A_1 \cup \cdots \cup A_r$ be the natural decomposition into simple index sets. For any resonance $R = (B,n)$, let $E_R \subset \operatorname{ad}(\gamma)$ be the set of all admissible families $J$ such that $R \in J$ and $R \not\subset R'$ for any other $R' \in J$; and let $E$ be the set of those families which do not contain any resonance of the form $R = (B,n)$. Then we have:
$$
\operatorname{ad}(\gamma) = E \cup \bigcup_{R = (B,n)} E_R,
$$ 
and this union is disjoint by construction. By definition of $\Omega_1(\delta,v)$, we have:
\begin{equation}
\label{e:product_omega_1}
\Omega_1(\delta,v) = (i \omega \cdot \gamma(n))^{-1} \big( \Omega_1( \delta,v/A_1) \cdots \Omega_1(\delta,v/A_r) - \Omega(\delta,v) \big).
\end{equation}
Since $n$ is not in the support of any resonance, then $(i \omega \cdot \gamma(n))^{-1} = (i \omega \cdot \gamma_J(n))^{-1}$ for any family $J \in \operatorname{ad}(\gamma)$. Notice also that\footnote{Let $E_1$, $E_2$ to sets of families of sets, the product $E_1 \times E_2$ is given by:
$$
E_1 \times E_2 = \{ J = J_1 \cup J_2 \, : \, J_1 \in E_1, \quad J_2 \in E_2 \}.
$$} 
$E = \operatorname{ad}(\gamma / A_1) \times \cdots \times \operatorname{ad}(\gamma / A_r)$. This shows that the set $E$ contributes with the product
$$
(i \omega \cdot \gamma(n))^{-1} \Omega_1( \delta,v/A_1) \cdots \Omega_1(\delta,v/A_r) 
$$
in \eqref{e:product_omega_1}. Similarly, for every resonance $R = (B,n)$, the set $E_R$ contributes with the product
$$
- (i \omega \cdot \gamma_J(n))^{-1}  \Omega_2(\delta,v/A\setminus A(B)) \Omega_1(\delta,v/A(b_1)) \cdots \Omega_1(\delta,v/A(b_s)),
$$
where $B = \{ b_1, \ldots, b_s \}$. Iterating this procedure covering the tree from the root $n$ towards its predecessors, we obtain the claim. The proof for $\Omega_2(\delta,v)$ is similar.
\end{proof}

In the above sum, however, there are in general many cancelations of signs. To avoid counting  summands which actually cancel each other out, in $\operatorname{ad}(\gamma)$ we define the following equivalence relation:

\begin{definition}
Let $J_1 = \{ R_1, \ldots, R_n \}$ and $J_2= \{Q_1, \ldots, Q_m \}$ belong to $\operatorname{ad}(\gamma)$ (resp. to $\operatorname{ad}^*(\gamma)$). We say that $J_1 \sim J_2$ if for every $Q_j = (C_j,d_j) \in J_1 \setminus J_2$ there exist $R_i = (B_i,a_i), R_k = (B_k,a_k) \in J_1$ such that $R_i \subset Q_j \subset R_k$, $a_k = d_j = a_i$, and $B_k \subset C_j \subset B_i$; and for every $R_j = (B_j,a_j) \in J_1 \setminus J_2$, there exist $Q_i = (C_i,d_i), Q_k = (C_k,d_k) \in J_2$ such that $Q_i \subset R_j \subset Q_k$, $d_k = a_j = d_i$, and $C_k \subset B_j \subset C_i$.

We denote $\operatorname{\mathbf{ad}}(\gamma) = \operatorname{ad}(\gamma) / \sim$  (resp. $\operatorname{\mathbf{ad}}^*(\gamma) = \operatorname{ad}^*(\gamma)/\sim$). We identify $[J]$ with $\mathbf{J}$, and define, for any equivalent class $[J] \in \operatorname{\mathbf{ad}}(\gamma)$ (resp. $\operatorname{\mathbf{ad}}^*(\gamma)$), the minimal element $\mathbf{J} \in [J]$ satisfying that $\mathbf{J} \subset J$ for every $J \in [J]$.
\end{definition}

\begin{figure}[h]
\label{f:resonances_2}
\includegraphics[scale=0.5]{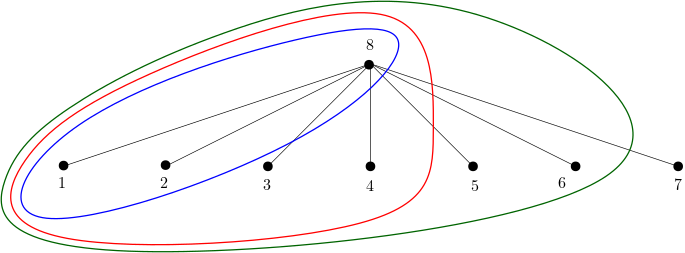}
\caption{Tree given by $\delta = (0, \ldots, 0, 7 )$. Let $\color{blue} R_1 = (B_1,a)$, $\color{red} R_2 = (B_2,a)$, $ \color{darkgreen} R_3 = (B_3,a)$ with $a= 8$, $B_1 = \{3,4,5,6,7 \}$, $B_2 = \{ 5,6,7 \}$, and $B_3 = \{7\}$. We have that $J_1 = \{R_1,R_3 \}$ and $J_2 = \{R_1,R_2,R_3 \}$ satisfy $J_1 \sim J_2$.}
\end{figure}

\begin{lemma}
\label{l:second_description}
The following two identities hold:
\begin{align}
\label{e:simplification}
\Omega_1(\delta,v) & = \sum_{ \mathbf{J} \in \operatorname{\mathbf{ad}}(\gamma)}   \chi(\mathbf{J}) \prod_{c \in A}  ( i \omega \cdot \gamma_{\mathbf{J}}(c))^{-1}, \\[0.2cm]
 \Omega_2(\delta,v) & = \sum_{ \mathbf{J} \in \operatorname{\mathbf{ad}}^*(\gamma)}   \chi(\mathbf{J}) \prod_{c \in A \setminus \{n \}}  ( i \omega \cdot \gamma_{\mathbf{J}}(c))^{-1},
\end{align}
for some coefficients $\chi(\mathbf{J}) \in \mathbb{Z}$ satisfying  $\vert \chi(\mathbf{J}) \vert \leq {4^n}$.
\end{lemma}

\begin{proof}
 The proof for $\Omega_2$ follows by the same argments as for $\Omega_1$, so we concentrate on the latter. By Lemma \ref{l:first_description}, we have
$$
\Omega_1(\delta,v) = \sum_{ J \in \operatorname{ad}(\gamma)} \prod_{c \in A} (-1)^{\# J} ( i \omega \cdot \gamma_{J}(c))^{-1}.
$$
Since $\sim$ is an equivalent relation in $\mathbf{ad}(\gamma)$, it remains to show that the following cancelation of signs holds: for every $[J] \in \operatorname{\mathbf{ad}}(\gamma)$,
$$
\sum_{ J \in [J]  } \prod_{c \in A} (-1)^{\# J} ( i \omega \cdot \gamma_{J}(c))^{-1} = \chi(\mathbf{J}) \prod_{c \in A}  ( i \omega \cdot \gamma_{\mathbf{J}}(c))^{-1},
$$
for certain explicit coefficient $\chi(\mathbf{J})$ given below. Assume first that $\mathbf{J} = \{ R_0, R^0 \}$ with $R_0 = (B_0,a)$, $R^0 = (B^0,a)$ and $B^0 \subset B_0$. We consider the natural decomposition of $\mathcal{B} = \mathcal{B}_{R^0} \setminus \mathcal{B}_{R_0}$ into simple index sets (see figure \ref{f:resonances_4}):
$$
 \mathcal{B} = A_1 \cup \cdots \cup A_p.
$$
By hypothesis on $R_0$ and $R^0$, we have:
$$
\sum_{c \in \mathcal{B}} v(c) = 0.
$$
We call $\mathcal{\tau} = \{ \mathcal{A}_1, \ldots, \mathcal{A}_q \}$ a \textbf{covering decomposition} of $\mathcal{B} = \mathcal{B}_{R^0} \setminus \mathcal{B}_{R_0}$ if $\mathcal{B} = \mathcal{A}_1 \cup \cdots \cup \mathcal{A}_q$ is a disjoint union of index sets $\mathcal{A}_i$, each of them decomposing into a disjoint union of simple index sets $A_j$. In particular,
\begin{equation}
\label{e:property_tau}
\sum_{c \in \mathcal{A}_i} v(c) = 0, \quad i = 1, \ldots, q.
\end{equation}
We call $\kappa(\tau) := q$ the degree of the covering decomposition $\tau$. Let $\tau$ be a covering decomposition of $\mathcal{B}$, we say that $\tau$ it is a \textbf{maximal covering decomposition} of $\mathcal{B}$ if moreover each $\mathcal{A}_i$ is maximal with respect to property \eqref{e:property_tau}, that is, $\mathcal{A}_i$ can not be decomposed itself into a non-trivial covering decomposition. In general, a maximal covering decomposition is not unique (see Figure \ref{f:resonances_4}).
Let 
$$
\mathscr{T} := \{ \tau_1, \ldots, \tau_N \}
$$ 
be the set of maximal covering decompositions of $\mathcal{B} = \mathcal{B}_{R^0} \setminus \mathcal{B}_{R_0}$. Assume first that $\mathscr{T} = \{ \tau_1 \}$. Let $\tau_1 = \{ \mathcal{A}_1 , \ldots , \mathcal{A}_{\kappa(\tau_1)} \}$. Let us split 
$$
[J] = \{\mathbf{J} \} \cup [J]_1 \cup \cdots \cup [J]_k,
$$
where $J' \in [J]_j$ if $\# J' = \# \mathbf{J} + j$. It turns out that $k = \kappa(\tau_1)-1$. Indeed, if 
$$
J' = \{ R_0 , R_1, \ldots, R_k ,R^0 \} \in [J]_k,
$$
with $R_0 \subset R_1 \subset \cdots \subset R_k \subset R^0$, then we have
$$
\mathcal{B}_{R_j} = \mathcal{B}_{R_0} \cup \mathcal{A}_{i_1} \cup \cdots \cup \mathcal{A}_{i_j}, \quad j = 1, \ldots, k,
$$
for some $\{i_1, \ldots, i_j \} \subset \{1, \ldots, \kappa(\tau_1) \}$. This means that we can write $[J]_k$ as a disjoint union:
$$
[J]_k = \bigcup_{ \{ l_1, \ldots, l_{k} \} \in I_k^k} \{ R_0, R_1^{l_1}, \ldots, R_{k}^{l_{k}}, R^0 \},
$$
where we say that $\{ l_1, \ldots, l_{k} \} \in I_k^k$ if by definition $R_1^{l_1} \subset \cdots \subset R_k^{l_k}$, where the index $l_\iota$ runs over the set  $\{1, \ldots, {k+1 \choose \iota} \}$ for all $\iota = 1, \ldots, k$. In other words, let $\mathcal{P}^{k+1}_\iota$ be the set of parts of $\{1, \ldots, k+1 \}$ of $\iota$ elements, for every $l_\iota \in \{1, \ldots, {k+1 \choose \iota} \}$, there exists a unique subset $\{i_1, \ldots, i_\iota\} \in \mathcal{P}^{k+1}_\iota$ such that
$$
\mathcal{B}_{R_\iota^{l_\iota}}  =  \mathcal{B}_{R_0} \cup \mathcal{A}_{i_1} \cup \cdots \cup \mathcal{A}_{i_\iota}.
$$
More generally, for $1 \leq j \leq k$,
$$
[J]_j = \bigcup_{ \{ s_1, \ldots, s_j \} \in S_j^k} \; \bigcup_{ \{ l_{s_1}, \ldots, l_{s_j} \} \in I_j^k(s_1, \ldots, s_j)} \{ R_0, R_{s_1}^{l_{s_1}}, \ldots, R_{s_j}^{l_{s_j}}, R^0 \},
$$
where  $\{ s_1, \ldots, s_j \} \in S_j^k$ if by definition $1 \leq s_1 < \cdots < s_j \leq k$, and we say that $\{ l_{s_1}, \ldots, l_{s_j} \} \in I_j^k(s_1, \ldots, s_j)$ if $R_{s_1}^{l_{s_1}} \subset \cdots \subset R_{s_j}^{l_{s_j}}$. 

Counting elements, we find $\# [J]_k = (k+1)!$, and more generally
\begin{align*}
\#[J]_j & = \sum_{\{ s_1, \ldots, s_j \} \in S_j^k} { k+1 \choose s_1} {k+1-s_1 \choose s_2 -s_1 } \cdots {k+1- s_{j-1} \choose s_j - s_{j-1} } \\[0.2cm]
 & = \sum_{\{ s_1, \ldots, s_j \} \in S_j^k} \frac{(k+1)!}{s_1! \cdots (s_j - s_{j-1})! (k+1-s_j)!}.
\end{align*}
Let us define the multinomial coefficients:
$$
{k+1 \choose i_1 \cdots i_j } := \frac{(k+1)!}{i_1 ! \cdots i_j (k+1 - i_1 - \cdots - i_j)!}.
$$
Since $\prod_{c \in A} ( i \omega \cdot \gamma_{J}(c))^{-1} = C_{[J]}$  is a constant for every $J \in [J]$, we obtain that 
\begin{equation}
\label{e:huge_cancelation}
 C_{[J]} \sum_{ J \in [J] } (-1)^{\# J} = C_{[J]} (-1)^{\# \mathbf{J}} \left( 1 +  \sum_{j=1}^k \sum_{\substack{ j \leq i_1+ \cdots + i_j \leq k \\ i_l \geq 1}} (-1)^{j} {k + 1 \choose i_1 \cdots i_j} \right) = C_{[J]} (-1)^{\# \mathbf{J} + k},
\end{equation}
by the multinomial theorem. Indeed, we have:
\begin{align*}
1 +  \sum_{j=1}^k \sum_{\substack{ j \leq i_1+ \cdots + i_j \leq k \\ i_l \geq 1}} (-1)^{j} {k + 1 \choose i_1 \cdots i_j} &  = 1 - \sum_{j=2}^{k+1} \sum_{i=1}^j (-1)^i i^{k+1} {j \choose i} \\[0.2cm]
 & =  -\sum_{i=1}^{k+1} \sum_{j=i}^{k+1} (-1)^{i} i^{k+1} { j \choose i } \\[0.2cm]
 & = - \sum_{i=1}^{k+1} (-1)^{i} i^{k+1} {k+2 \choose i+1} \\[0.2cm]
 & = (-1)^k,
\end{align*}
where the last equality holds by finite differencies\footnote{Let $P$ be any polynomial of degree less than $n$, then:
$$
\sum_{j=0}^n (-1)^j {n \choose j } P(j) = 0.
$$}. Defining $\chi(\mathbf{J}) := (-1)^{ \# \mathbf{J} + k}$, the claim holds in this case. 

Let us consider now the case in which $\mathbf{J} = \{R_0,R^0 \}$ as before but in this case $\mathscr{T} = \{ \tau_1, \ldots, \tau_N \}$ with $ N \geq 2$. Let $\tau, \tau'$ be two covering decompositions of $\mathcal{B} = \mathcal{B}_{R^0} \setminus \mathcal{B}_{R_0}$ (not necessarily maximal). We say that $\tau \subseteq \tau'$ if for each $\mathcal{A} \in \tau$ there exists $\mathcal{A}' \in \tau'$ such that $\mathcal{A} \subset \mathcal{A}'$. Let $\tau, \tau'$ be two covering decompositions of $\mathcal{B}$. We define
$$
\tau \wedge \tau'
$$
as the unique covering decomposition of $\mathcal{B}$ such that $\tau,\tau' \subseteq \tau \wedge \tau'$ and there is no other covering decomposition $\tau''$ verifying $\tau,\tau' \subseteq \tau'' \subsetneq \tau \wedge \tau'$.
\begin{figure}[h]
\includegraphics[scale=0.5]{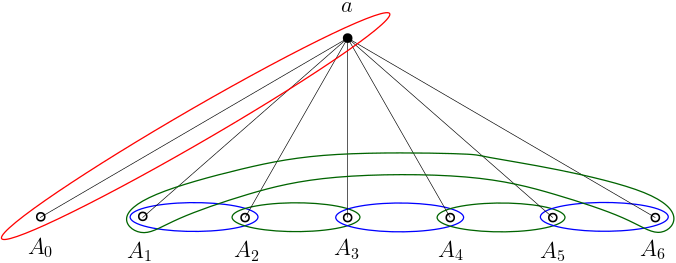}
\caption{We represent here $\mathbf{J} = \{ R_0, R^0 \}$ with $\mathcal{B}_{R_0} = \{a\} \cup A_0$ and $\mathcal{B}_{R^0} = \{ a \} \cup A_0 \cup A_1 \cup \cdots \cup A_6$. There are two maximal covering decompositions of $\mathcal{B} = \mathcal{B}_{R^0} \setminus \mathcal{B}_{R_0}$ given by $\color{blue} \tau_1 = \{ \mathcal{A}_1, \mathcal{A}_2, \mathcal{A}_3 \}$ and $\color{darkgreen} \tau_2 = \{\mathcal{A}_4, \mathcal{A}_5, \mathcal{A}_6 \}$, where $\mathcal{A}_1 = A_1 \cup A_2$, $\mathcal{A}_2 = A_3 \cup A_4$, $\mathcal{A}_3 = \{ A_5,A_6 \}$, $\mathcal{A}_4 = \{A_1,A_6\}$, $\mathcal{A}_5 = \{A_2,A_3 \}$, and $\mathcal{A}_6 = \{A_4,A_5 \}$. In this case $\tau_1 \wedge \tau_2 = \{ \mathcal{B} \}$.}
\label{f:resonances_4}
\end{figure}
We  define the closure $\overline{\mathscr{T}}$ of the set $\mathscr{T}$ of minimal covering decompositions by:
$$
\overline{\mathscr{T}} := \bigcup_{j=1}^N  \bigcup_{\{i_1, \ldots, i_j \} \in \mathcal{P}^N_j} \{ \tau_{i_1} \wedge \cdots \wedge \tau_{i_j} \}.
$$
We rewrite this set by labelling its elements as $\overline{\mathscr{T}} = \{ \tau_1, \ldots, \tau_M \}$, where $N < M$. Since $\prod_{c \in A} ( i \omega \cdot \gamma_{J}(c))^{-1} = C_{[J]}$, using \eqref{e:huge_cancelation} repeteadly we obtain that \eqref{e:simplification} holds provided that
\begin{align*}
\chi(\mathbf{J}) & := \sum_{j=1}^N \sum_{\{i_1, \ldots, i_j \} \in \mathcal{P}^N_j} (-1)^{\# \mathbf{J} + \kappa(\tau_{i_1} \wedge \cdots \wedge \tau_{i_j})+j}.
\end{align*}
It remains to show that $\vert \chi(\mathbf{J}) \vert \leq 2^n$. To this aim, we rewrite
\begin{equation}
\label{e:chi_expression}
 \chi(\mathbf{J}) = \sum_{l=1}^M (-1)^{\# \mathbf{J} + \kappa(\tau_l)} \lambda(\tau_l),
\end{equation}
where the coefficients $\lambda(\tau_l)$ are given by
\begin{equation}
\label{e:lambdas}
\lambda(\tau_l) = \sum_{j=1}^N \sum_{ \tau_{i_1} \wedge \cdots \wedge \tau_{i_j} = \tau_l} (-1)^{j}, \quad l = 1, \ldots, M.
\end{equation}

\noindent  We next notice that on $\overline{\mathscr{T}}$ one can define a partial ordering $\prec$ so that $\tau_i \prec \tau_j$ if $\tau_i \subset \tau_j$. This ordering has the unique maximal element $\tau_M$. Moreover, for every $\tau_i \in \overline{\mathscr{T}}$, we have
\begin{equation}
\label{e:cancelation_of_signs}
\sum_{ j \, : \, \tau_j \preceq \tau_i } \lambda(\tau_j) = \sum_{j = 1}^{n(\tau_i)} (-1)^j {n(\tau_i) \choose j} = (-1)^{n(\tau_i)},
\end{equation}
where $n(\tau_i) = \#  \{ j \leq N \, : \, \tau_j \preceq \tau_i \}$. Let $A$ be the adjacency matrix of the directed graph given by the order $\prec$, that is, $A_{ij} = 1$ if $ \tau_i \prec \tau_j$ and there is no $l \notin \{i,j \}$ such that $\tau_i \prec \tau_l \prec \tau_j$, and zero otherwise. Let $\vec{\lambda} := ( \lambda(\tau_1), \ldots, \lambda(\tau_M))$ and $\vec{I} := ((-1)^{n(\tau_1)}, \ldots, (-1)^{n(\tau_M)})$. Condition \eqref{e:cancelation_of_signs} means that
$$
(\operatorname{Id} + A + \cdots + A^M) \vec{\lambda} = \vec{I}.
$$
Since $A$ is a nilpotent matrix and $A^M = 0$, this implies that $\vec{\lambda} = (I-A) \vec{I}$. Then, using that $M \leq 2^n$ we obtain $\vert \chi(\mathbf{J}) \vert \leq 4^n$. 
\medskip

In the general case, notice that 
$\mathbf{J} \in \mathbf{ad}(\gamma)$ can be decomposed in a unique way as:
$$
\mathbf{J} = \mathbf{G}_1 \cup \cdots \cup \mathbf{G}_L, \quad L \leq 2^n,
$$
where $\# \mathbf{G}_l  \leq 2$ and if $R_0(l),R^0(l) \in \mathbf{G}_l$, then $R_0(l)=(B_0,a)$, $R^0(l) = (B^0,a)$ and $B^0 \subset B_0$. Then, we can write
$$
[J] = \bigotimes_{l=1}^L \big( \mathbf{J} \cup [J]_1^l \cup \cdots \cup [J]_{k(l)}^l \big),
$$
where we define $[J]_1^l$ as before but only counting resonances that are in between $R_0(l)$ and $R^0(l)$. Following the previous argument separately for every $l \in \{1, \ldots, L\}$ and taking the product, the result holds.
\end{proof}

\subsection{Convergence of the Lindstedt series}

In this section we prove the convergence of the Lindstedt series, by reducing it to \cite[Prop. 3]{Elia89}.

\begin{teor} 
\label{t:bound_lindstedt_series}
Let $0 < \sigma < s_0$. Then, there exists $C = C(s_0,\sigma) > 0$ such that
\begin{align*}
\Vert H_n \Vert_{s_0-\sigma}  \leq  C^n \Vert V \Vert_{s_0}^n; \quad 
\Vert R'_n \Vert_{s_0-\sigma}  \leq C^n \Vert V \Vert_{s_0}^n.
\end{align*}
In particular,  if $\epsilon = \epsilon(V,\omega) > 0$ is sufficiently small, then $\Vert H \Vert_{s,\epsilon} < \infty$ and $\Vert R' \Vert_{s,\epsilon} < \infty$.
\end{teor}

The main part of the proof of the convergence of the Lindstedt series lies in \cite[Prop. 3]{Elia89}. 
We define the coefficients
\begin{equation}
\label{e:rho}
\rho(\delta,v) := \sum_{\mathbf{J} \in \operatorname{\mathbf{ad}}(\gamma)} \vert \chi(\mathbf{J}) \vert; \quad \rho^*(\delta,v) := \sum_{\mathbf{J} \in \operatorname{\mathbf{ad}}^*(\gamma)} \vert \chi(\mathbf{J}) \vert
\end{equation}
In \cite{Elia89}, the estimate of the coefficients $\Omega_j$, for $j = 1,2$, is reduced to the estimate of $\rho, \rho^*$ by using a generalization of several lemmas of Siegel. 

\begin{lemma}\label{l:lemma_eliasson} Let $\delta \in \Delta(n)$. Then:
\begin{align}
\label{e:bound_omega_1}
\vert \Omega_1(\delta,v) \vert & \leq \rho(\delta,v) (2^{4 \gamma + 3} \varsigma )^n \prod_{v_j \neq 0} \vert v_j \vert^{3 \gamma}; \\[0.2cm]
\label{e:bound_omega_2}
\vert \Omega_2(\delta,v) \vert &  \leq \rho^*(\delta,v) ( 2^{4 \gamma + 3} \varsigma )^n \prod_{\substack{v_j \neq 0 \\ j \neq n}} \vert v_j \vert^{3 \tau},
\end{align}
where $\gamma$ and $\varsigma$ are given by \eqref{e:diophantine}.
\end{lemma}

\begin{remark}
 Lemma \ref{l:lemma_eliasson} is essentially equivalent to \cite[Prop. 3]{Elia89}. However, in \cite{Elia89} there is not any explicit reference to the coefficient $\rho$ given by \eqref{e:rho}, but to the number of summands in \eqref{e:first_omega}, that is, in this reference $\rho(\delta,v) = \# \operatorname{ad}(\gamma)$.  
However, by the cancelation of signs given by Lemma \ref{l:second_description}, the number of non-vanishing terms in \eqref{e:first_omega} is bounded by \eqref{e:rho}.
\end{remark} 
\begin{proof}
 We prove Lemma \ref{l:lemma_eliasson} from \cite[Prop. 2]{Elia89} by adapting the proof of \cite[Prop. 3]{Elia89}. We proceed by induction. The case $n = 1$ holds so we assume that it holds for $n-1$. By the recurrence relations described in Theorem \ref{t:main_lindstedt_teor}, the estimate of $\Omega_2$ follows from that of $\Omega_1$, so we concentrate on $\Omega_1$. Let $\delta \in \Delta(n)$, let us consider the family
 $$
 \mathcal{E}_1(\delta) = \{ (B,a): B = \{ b _1, \ldots, b_s \}  \},
 $$ 
such that $b_1, \ldots, b_s$ are pairwise unrelated in $A(a) \setminus \{a \}$. Let $\mathcal{E}_0(\delta)$ the family of pairs $(B,a)$ as before with $\# B = 1$. We have $\mathcal{E}_0(\delta) \subset \mathcal{E}_1(\delta)$. Let $\mathcal{E}(\delta)$ be such that $\mathcal{E}_0(\delta) \subset \mathcal{E}(\delta) \subset \mathcal{E}_1(\delta)$,
then one can define coefficients $\Omega_1(\delta,v; \mathcal{E})$ and $\Omega_2(\delta,v; \mathcal{E})$ by changing $\Gamma_1(\delta)$ into 
$$
\Gamma(\delta) = \{ B \, : \, (B,n) \in \mathcal{E}(\delta) \}
$$ 
in the definition \eqref{e:definition_omega} of $\Omega$. Similarly, one can define $\operatorname{ad}(\gamma,\mathcal{E})$ and $\operatorname{\mathbf{ad}}(\gamma,\mathcal{E})$ (and similarly $\operatorname{ad}^*(\gamma,\mathcal{E})$ and $\operatorname{\mathbf{ad}}^*(\gamma,\mathcal{E})$) by considering only resonances belonging to $\mathcal{E}$. Let $A' \subset A$ be such that $\delta/A'$ is a simple index set, we define $\mathcal{E}(\delta/A')$ as the family of pairs $(B,a) \in \mathcal{E}(\delta)$ such that $a \in A'$ and $B = \{b_1, \ldots, b_s \} \subset A'$ consists of pairwise unrelated elements in $A(a) \setminus \{a\}$. 

Let $(B_1,a) \in \mathcal{E}(\delta) \setminus \mathcal{E}_0(\delta)$ be such that for no $(B_2,a) \in \mathcal{E}(\delta) \setminus \mathcal{E}_0(\delta)$, $A(a)\setminus A(B_2) \supset A(A) \setminus A(B_1)$ (such $(B_1,a)$ always exists). Set $\mathcal{E}'(\delta) = \mathcal{E}(\delta) \setminus \{ (B_1,a) \}$. Then:
\begin{equation}
\label{e:omega_decomposition}
\Omega_1(\delta, v; \mathcal{E}) = \Omega_1(\delta,v; \mathcal{E}') - \Omega_2(\delta,v /A(a) \setminus A(B_1); \mathcal{E}') \Omega_1(\delta',v'; \mathcal{E}'),
\end{equation}
where $\delta' : A \setminus (A(a) \setminus A(B_1)) \cup \{ a \} \to \mathbb{Z}$ is defined by $\delta/A \setminus (A(a) \setminus A(B_1))$ on $A \setminus (A(a) \setminus A(B_1))$ and $\delta'(a) = \# B_1$; and $v' : A \setminus (A(a) \setminus A(B_1)) \cup \{ a \} \to \mathbb{Z}^d$ is defined to be equal to $v$ everywhere except at $a$ where $v'(a) = 0$. Notice that $\Omega_2(\delta,v /A(a) \setminus A(B_1); \mathcal{E}')$ is non vanishing only if $(B_1,a)$ is a resonance. Denoting 
$$
\gamma_{\mathbf{J}} := \prod_{c \in A} ( i \omega \cdot \gamma_{\mathbf{J}}(c))^{-1}, \quad \gamma^*_{\mathbf{J}}:= \prod_{c \in A \setminus \{n \}} ( i \omega \cdot \gamma_{\mathbf{J}}(c))^{-1},
$$
using \eqref{e:omega_decomposition} and Lemma \ref{l:second_description}, this gives:
\begin{align*}
\sum_{ \mathbf{J} \in \operatorname{\mathbf{ad}}(\gamma,\mathcal{E})}   \chi(\mathbf{J}) \gamma_{\mathbf{J}} & = \sum_{ \mathbf{J} \in \operatorname{\mathbf{ad}}(\gamma,\mathcal{E}')}   \chi(\mathbf{J}) \gamma_{\mathbf{J}} - \sum_{ \mathbf{J} \in \operatorname{\mathbf{ad}}(\gamma,\mathcal{E}) \setminus \mathbf{ad}(\gamma, \mathcal{E}')}   \chi(\mathbf{J}) \gamma_{\mathbf{J}}  \\[0.2cm]
 & =  \sum_{ \mathbf{J} \in \operatorname{\mathbf{ad}}(\gamma,\mathcal{E}')}   \chi(\mathbf{J}) \gamma_{\mathbf{J}} -  \left( \sum_{ \mathbf{J}_1 \in \operatorname{\mathbf{ad}}^*(\gamma_1,\mathcal{E}')}   \chi(\mathbf{J}_1) \gamma^*_{\mathbf{J}_1}   \right) \left( \sum_{ \mathbf{J}_2 \in \operatorname{\mathbf{ad}}(\gamma_2,\mathcal{E}')}   \chi(\mathbf{J}_2) \gamma_{\mathbf{J}_2} \right),
\end{align*}
where $\gamma_1 = \gamma/ A(a) \setminus A(B_1) $ and $\gamma_2 = \gamma_{\delta',v'}$. Moreover, notice that 
$$
\mathbf{J} \in \operatorname{\mathbf{ad}}(\gamma,\mathcal{E}) \setminus \mathbf{ad}(\gamma, \mathcal{E}')
$$
if and only if $\mathbf{J} \in \operatorname{\mathbf{ad}}(\gamma,\mathcal{E})$ and $(B_1,a) \in \mathbf{J}$. Then, for every $ \mathbf{J} \in \operatorname{\mathbf{ad}}(\gamma,\mathcal{E}) \setminus \mathbf{ad}(\gamma, \mathcal{E}')$ there exist unique $\mathbf{J}_1 \in \operatorname{\mathbf{ad}}^*(\gamma_1,\mathcal{E}')$ and $\mathbf{J}_2 \in \operatorname{\mathbf{ad}}(\gamma_2,\mathcal{E}')$ such that $\mathbf{J} = \mathbf{J}_1 \cup \mathbf{J}_2$. 
Finally, observing that $ \vert \chi(\mathbf{J}) \vert = \vert \chi(\mathbf{J}_1) \vert \vert \chi(\mathbf{J}_2) \vert$, the claim holds by iterating this decomposition and using the induction hypothesis.
\end{proof}

After \cite[Prop. 3]{Elia89}, the proof of the convergence of the Lindstedt series in \cite{Elia89} is not straightforward since it is not provided any estimate on the number $\rho(\delta,v)$. Actually, the cancelation of signs given here by Lemma \ref{l:second_description} is not exploited in \cite{Elia89}, which makes the end of the proof more involved. Alternatively, we prove the following:

\begin{prop} 
\label{p:combinatorial_lemma}
There exists a universal constant $C > 0$ such that, for every $\delta \in \Delta(n)$ and $v \in (\mathbb{Z}^d)^n$, $\rho(\delta,v) \leq C^{n}$ and $ \rho^*(\delta,v) \leq C^n$. 
\end{prop}

\begin{remark}
We prove Proposition \ref{p:combinatorial_lemma} with  $C = 32$, although this constant is not shown to be sharp. A particular case of this proposition is given in \cite[Lemma 15]{Eliasson96} with $C = 2$.
\end{remark}

\begin{proof}

We prove only the inequality $\rho(\delta,v) \leq C^n$. The case for $\rho^*$ follows from this. Let $A = \{1, \ldots, n \}$, we define the set $\mathbf{B}(A)$ of \textit{free resonances} in $A$ by saying that $\mathcal{B} \in \mathbf{B}(A)$ if $\mathcal{B} \subset A$ and $\sum_{b \in \mathcal{B}} v(b) = 0$. Let $\iota: (A, \delta) \to A$ be the natural identification. We have the trivial bound $\# \mathbf{B}(A) \leq 2^n$.

We next define the following map $T:  \operatorname{\mathbf{ad}}(\gamma) \to \mathbf{B}(A)$. Let $\mathbf{J} \in \operatorname{\mathbf{ad}}(\gamma)$, we decompose 
\begin{equation}
\label{e:decomposition_family}
\mathbf{J} = \mathbf{I}_1 \cup \cdots \cup \mathbf{I}_k,
\end{equation}
where $\mathbf{I}_1$ is the set of maximal\footnote{Let $\mathbf{J} \in \mathbf{ad}(\gamma)$, a resonance $R \in \mathbf{J}$ is said to be maximal if there is no resonance $R' \in \mathbf{J}$ such that $R \varsubsetneq R'$.} resonances in $\mathbf{J}$, $\mathbf{I}_2$ is the set of maximal resonances in $\mathbf{J} \setminus \mathbf{I}_1$ and so on. Let $\mathbf{i}_1$ be the set of minimal resonances in $\mathbf{J}$. We set:
$$
T( \mathbf{J}) := \iota \Big( \Big( \Big( \bigcup_{R \in \mathbf{I}_1} \mathcal{B}_R  \setminus  \bigcup_{R \in \mathbf{I}_2} \mathcal{B}_R \Big) \cup \bigcup_{R \in \mathbf{I}_3} \mathcal{B}_R \Big) \setminus \cdots \bigcup_{R \in \mathbf{I}_k} \mathcal{B}_R \Big) .
$$
We also define the map $\mathcal{T} : \operatorname{\mathbf{ad}}(\gamma) \to \mathbf{B}(A) \times \mathbf{B}(A) \times \mathbf{B}(A)$ by
$$
\mathcal{T}( \mathbf{J}) := \Big( T(\mathbf{J}) , \iota \Big( \bigcup_{R \in \mathbf{i}_1} \mathcal{B}_R \Big) , \iota \Big( \bigcup_{R \in \mathbf{I}_1} \mathcal{B}_R \Big) \Big).
$$ 

\begin{lemma}
\label{l:key_intersection_lemma}
Let $\mathbf{J} \in \operatorname{\mathbf{ad}}(\gamma)$ and $R \in \mathbf{J}$. Assume that $\mathcal{B}_R = \mathcal{B}_{R_1} \cup \mathcal{B}_{R_2}$ with $R_1,R_2 \subset R$ disjoint. Then:
\begin{enumerate}
\item $R_1, R_2 \notin \mathbf{J}$.

\item There exists $R' \in \mathbf{J}$ such that $R' \subset R$, and $R'$ overlaps with $R_1$ and $R_2$. In particular, $R \notin \mathbf{i}_1$.
\end{enumerate}
\end{lemma}

\begin{proof}
We can assume without loss of generality that $R_1 = (B_1,a_1)$, $R_2 = (B_2,a_2)$ and $a_2 \in B_1$. Then $a_2 \in \operatorname{supp} \mathbf{J}$. Since we have $\sum_{b \in A(a_2) \setminus A(B_2)} v(b) = 0$, then $\gamma_{\mathbf{J}}(a_2) \neq \sum_{b \in A(a_2) \setminus A(B_2)} v(b)$, hence there exists another resonance $R' \subsetneq R$, $R' \in \mathbf{J}$, such that $a_2 \in \operatorname{supp} R'$. This implies that $R'$ overlaps with $R_1$ and $R_2$. Then, clearly $R_1$ and $R_2$ can not belong to $\mathbf{J}$. This finishes the proof.

\end{proof}

\begin{lemma} 
Let $(B,\delta/B) \subset (A, \delta)$ be a simple index set such that $\sum_{b \in B} v(b) = 0$. Then there exists a unique maximal decomposition of the form $B = \mathcal{B}_{R_1} \cup \cdots \cup \mathcal{B}_{R_k}$ where $R_j$ are disjoint resonances. By maximal we mean that each $R_j$ can not be decomposed in a non trivial union of disjoint resonances. We call the set $\{R_1, \ldots, R_k \}$ the \textbf{maximal covering decomposition} of $(B, \delta/B)$. We can generalize this definition to index sets by considering separately each connected component. 
\end{lemma}

\begin{remark}
Notice that we use the same terminology for maximal covering decompositions of simple index sets $(B,\delta/B)$ and for  maximal covering decompositions of $\mathcal{B} = \mathcal{B}_{R^0} \setminus \mathcal{B}_{R_0}$ in the proof of Lemma \ref{l:second_description}, which are different objects. We apology for this redundancy.
\end{remark}

\begin{proof}
Assume that there exist two different maximal decompositions in disjoint resonances
$$
B = \mathcal{B}_{R_1} \cup \cdots \mathcal{B}_{R_k} = \mathcal{B}_{R_1'} \cup \cdots \cup \mathcal{B}_{R_m'}.
$$
Notice that we can define a partial order in $\{R_1, \ldots, R_k \}$ (respectively in $\{R_1', \ldots, R_m' \}$) in the following way:
$$
R_p = (B_p, a_p) \prec R_q = (B_q,a_q) \Longleftrightarrow a_p \preceq b_q \text{ for some } b_q \in B_q.
$$
Let us assume that $R_1$ is a minimal resonance for this order. We claim that there exists $R_\iota'$ (say $R_1'$) such that $R_1 \subset R_1'$ or $R_1' \subset R_1$. Indeed, let $R_1' = (B_1', a_1')$ be the minimal resonance of the set $\{R_1', \ldots, R_m' \}$ intersecting $R_1$. If $R_1'$ overlaps with $R_1$, let  $a'_2 \in \mathcal{B}_{R_1} \setminus \mathcal{B}_{R_1'}$ be the maximal element for the order $\delta/ \mathcal{B}_{R_1} \setminus \mathcal{B}_{R_1'}$. Then there exists $b_1' \in B_1'$ such that $a_2' \preceq b_1'$. This implies that there exists another resonance $R_2'$ such that $R_2' = (B_2',a_2')$, and then $ R_2' \prec R_1'$, but this is a contradiction since $R_1'$ is the minimal resonance of the set $\{R_1', \ldots, R_m' \}$ intersecting $R_1$. Assume now that $R_1 \subset R_1'$. If this inclusion is strict, then there exists another resonance $S_1$ defined by $\mathcal{B}_{S_1} =\mathcal{B}_{R'_1} \setminus \mathcal{B}_{R_1}$,  hence $R_1 \cap S_1 = \emptyset$ and $R_1' = S_1 \cup R_1$. But then the decomposition $\{R_1', \ldots, R_m'\}$ is not maximal. Therefore $S_1 = \emptyset$, that is, $R_1 = R_1'$. If Otherwise $R_1' \subset R_1$, we obtain again that $R_1' = R_1$ by interchanging the roles of $R_1$ and $R_1'$. Finally, we can iterate the same argument for the sets $\{ R_2, \ldots, R_k \}$ and $\{R_2', \ldots, R_m' \}$.  This concludes the proof.
\end{proof}

\begin{lemma}
\label{l:injective}
The map $\mathcal{T}$ is injective.
\end{lemma} 

\begin{figure}[h]
\label{f:resonances_3}
\includegraphics[scale=0.6]{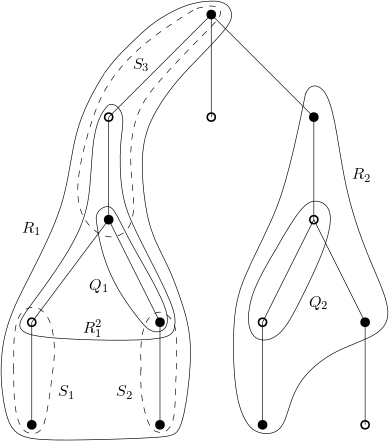}
\caption{In this example, $\mathbf{J} = \{Q_1,R_1^2,R_1, Q_2, R_2 \}$, the dashed resonances $S_1,S_2,S_3 \notin \mathbf{J}$, the set $\iota^{-1}(T(\mathbf{J}))$ is given by the black points, $\mathbf{I}_1 = \{R_1, R_2 \}$, and $\mathbf{i}_1 = \{Q_1 , Q_2 \}$. Notice that $\mathcal{B}_{R_1} = \mathcal{B}_{S_1} \cup \mathcal{B}_{S_2} \cup \mathcal{B}_{S_3}$ is a maximal covering decomposition of $\mathcal{B}_{R_1}$.}
\end{figure}

\begin{proof}
Let $(\mathcal{R}_1, \mathcal{R}_2, \mathcal{R}_3) \in \operatorname{Im}(\mathcal{T})$. Let $\mathbf{J} \in \operatorname{\mathbf{ad}}(\gamma)$ such that $\mathcal{T}(\mathbf{J}) = (\mathcal{R}_1, \mathcal{R}_2, \mathcal{R}_3)$. We have that
$$
\mathbf{I}_1 = \{ R_1, \ldots, R_r \}, \quad \mathbf{i}_1 = \{Q_1, \ldots, Q_r \},
$$
where $Q_j \subset R_j$ for all $1 \leq j \leq r$. Let $(B, \delta/B) = \iota^{-1}(\mathcal{R}_2)$. We decompose the set $B$ as
$$
B = \mathcal{B}_{Q_1} \cup \cdots \cup \mathcal{B}_{Q_r}
$$
in maximal covering decomposition (this finds $Q_1, \ldots, Q_r$ by Lemma \ref{l:key_intersection_lemma} and hence $\mathbf{i}_1$). Let us define $\mathbf{i}_2$ to be the set of minimal resonances in $\mathbf{J} \setminus \mathbf{i}_1$, and so on, so that
$$
\mathbf{J}= \mathbf{i}_1 \cup \cdots \cup \mathbf{i}_k,
$$ 
where notice that $k$ is the same as the one obtained by \eqref{e:decomposition_family}. We decompose:
\begin{equation}
\label{e:decomposition_family_in_minimal}
\mathbf{i}_l = \{ R^l_1, \ldots, R^l_{r_l} \}, \quad 1 \leq l \leq r,
\end{equation}
where $r_1 = r$, and $1 \leq r_l \leq r_{l-1}$ for $2 \leq l \leq r$. We have $R^1_j = Q_j$ and for any $r_{l-1} < j \leq r_l$, $R^{l-1}_j = R_j$, so that, for any $1 \leq j \leq r$,
$$
Q_j = R^1_j \subset R^2_j \subset \cdots \subset R_j.
$$
Next consider the maximal covering decomposition of $B' = \iota^{-1}(\mathcal{R}_3)$:
$$
B' = \mathcal{B}_{S_1} \cup \cdots \cup \mathcal{B}_{S_s},
$$
where $r \leq s$. Notice that at most one $Q_j$ can intersect a given $S_i$. Indeed, if $R_j \supset S_i$, then only $Q_j$ can intersect $S_i$.

We asign now two colors to the points of $(B',\delta/B')$. The set of black points is given by $\iota^{-1}(\mathcal{R}_1)$ and the set of white points is given by $B' \setminus \iota^{-1}(\mathcal{R}_1)$. We start with the resonance $Q_1$ and define $\mathcal{U}_1 := \mathcal{B}_{Q_1} \cup \mathcal{B}_1$, where $\mathcal{U}_1$ is a simple index set in $(B',\delta/B')$ and $\mathcal{B}_1$ is given only by points with different color of that of $\mathcal{B}_{Q_1}$ (possibly $\mathcal{B}_1 = \emptyset$). If $\mathcal{B}_1 = \emptyset$, then we define $R_1 = Q_1$. If now $\mathcal{B}_1$ is given by white points, we set $\mathcal{U}_1' = \mathcal{U}_1$. Otherwise, if $\mathcal{B}_1$ is given by black points, let $\mathscr{S}_1 := \{S_{j_1}, \ldots, S_{j_l} \}$ be such that $\mathcal{B}_{S_{j_\iota}} \cap \mathcal{B}_1 \neq \emptyset$, $\mathcal{B}_{S_{j_\iota}} \cap \mathcal{B}_{Q_1} = \emptyset$, and observe that there exists another $Q_i$ such that $\mathcal{B}_{Q_i} \cap \mathcal{B}_{S_{j_\iota}} \neq \emptyset$ (possibly $\mathscr{S}_1 = \emptyset$). Then $S_{j_\iota}$ can not intersect $R_1$ (otherwise, $S_{j_\iota}$ could not be contained in $R_i$, but this is the case since  $\mathcal{B}_{Q_i} \cap \mathcal{B}_{S_{j_\iota}} \neq \emptyset$). 

We next define $\mathcal{U}_1' := \mathcal{U}_1 \setminus ( \mathcal{B}_{S_{j_1}} \cup \cdots \cup \mathcal{B}_{S_{j_l}})$. We set $\mathcal{U}_2 = \mathcal{U}_1' \cup \mathcal{B}_2$ where $\mathcal{U}_2$ is a simple index set in $(B',\delta/B')$ and $\mathcal{B}_2$ is given by points of different color of that of $\mathcal{B}_1$. Notice that there exists only one resonance $R_1' \in \mathbf{J}$ such that $Q_1 \subset R_1' \subset \mathcal{U}_1'$. Indeed, if two resonances $R_1', R_1''$ satisfies this condition, then $J' = \mathbf{J} \setminus \{R_1' \}$ satisfies $J' \sim \mathbf{J}$, but this is a contradiction.  If $\mathcal{B}_2 = \emptyset$, then we define $\mathcal{B}_{R_1} := \mathcal{U}_1'$. Otherwise, if $\mathcal{B}_2 \neq \emptyset$, we define $R_1^2$ by $\mathcal{B}_{R_1^2} = \mathcal{U}_1'$. If now $\mathcal{B}_2$ is given by white points, we set $R_1^2$ by $\mathcal{B}_{R_1^2} = \mathcal{U}_2$. Otherwise, if $\mathcal{B}_2$ is given by black points, let $\mathscr{S}_2 := \{S_{j_1}, \ldots, S_{j_l} \}$ such that $\mathcal{B}_{S_{j_\iota}} \cap \mathcal{B}_2 \neq \emptyset$, $\mathcal{B}_{S_{j_\iota}} \cap \mathcal{U}_1' = \emptyset$. 

We next define $\mathcal{U}_2' := \mathcal{U}_2 \setminus ( \mathcal{B}_{S_{j_1}} \cup \cdots \cup \mathcal{B}_{S_{j_l}})$ and continue with this process.  We repeat this construction with all the minimal resonances $Q_j$. This finds $\mathbf{i}_2$. Repeating this process now starting from $\mathbf{i}_2$, we find $\mathbf{i}_3$, and so on. Using decomposition \eqref{e:decomposition_family_in_minimal}, we find $\mathbf{J}$ and this concludes the proof. We remark that if we define the map $\mathcal{T}$ on $[ \mathbf{J} ]$ then in general it is not injective. Indeed, if $J \sim \mathbf{J}$ (hence $\mathbf{J} \subset J$), then $\mathcal{T}(J) = \mathcal{T}(\mathbf{J})$.

\end{proof}

Therefore, by Lemma \ref{l:injective}, we get  $\# \operatorname{\mathbf{ad}}(\gamma) \leq \# ( \mathbf{B}(A) \times \mathbf{B}(A) \times \mathbf{B}(A)) \leq 2^{3n}$. This and Lemma \ref{l:second_description} conclude the proof of Proposition \ref{p:combinatorial_lemma}  with $C = 32$. 

\end{proof}

After this, it remains only to bound the function 
$$
\mathcal{G}_\hbar(\delta,x,\xi) := \sum_{v \in (\mathbb{Z}^d)^n} \int_{(\R^d)^n} \prod_{v_j \neq 0} \vert v_j \vert^{3 \tau} \widehat{\mathcal{F}}_\hbar(\delta,v,\eta) e^{i (v_1+ \cdots + v_n) \cdot x} e^{i (\eta_1 + \cdots + \eta_n) \cdot \xi} d\eta.
$$
To do so, we prove the following:

\begin{lemma}
\label{e:analytical_part}
Let $\delta \in \Delta(n)$ with $n \geq 2$ and $\sigma > 0$. Then there exists a universal constant $C > 0$ such that
$$
\mathbf{c}(\delta) \sup_{\mathbf{w} \in (\mathcal{Z}^d)^n}  \Big \lbrace \vert \sigma_\hbar(\mathbf{w},\delta) \vert e^{-\sigma (\vert w_1  + \cdots +  w_n \vert)} \Big \rbrace  \leq C^{n-1} \left(\frac{(n-1)^{n-1}}{(n-1)!} \right)^2 \left( \frac{1}{e \sigma} \right)^{2(n-1)}.
 $$
\end{lemma}

By Stirling's approximation and the trivial case $n =1$, Lemma \ref{e:analytical_part} implies  immediately the following:
\begin{corol}
Let $\delta \in \Delta(n)$ with $n \geq 1$ and $\sigma > 0$. Then there exists a constant $C_\sigma > 0$ such that
$$
\mathbf{c}(\delta) \sup_{\mathbf{w} \in (\mathcal{Z}^d)^n}  \Big \lbrace \vert \sigma_\hbar(\mathbf{w},\delta) \vert e^{-\sigma (\vert w_1  + \cdots +  w_n \vert)} \Big \rbrace  \leq  C_\sigma^{n-1}.
$$
\end{corol}

\begin{proof}
We proceed by induction. The case $n=2$ is trivial. Now, by \eqref{e:graph_commutator}, we have
\begin{align*}
 \sigma_\hbar(\mathbf{w}, \delta) & = \sigma_\hbar^{r}(w_n, \Sigma_1(\mathbf{w}), \ldots, \Sigma_r(\mathbf{w})) \sigma_\hbar(\mathbf{w}/A_1,\delta/A_1) \cdots   \sigma_\hbar(\mathbf{w}/A_r,\delta/A_r) \\[0.2cm]
  & = \sigma_\hbar^1( \Sigma_1(\mathbf{w}), w_n + \Sigma_2(\mathbf{w}) + \cdots + \Sigma_r(\mathbf{w}))\sigma_\hbar^{r-1}(w_n, \Sigma_2(\mathbf{w}), \ldots, \Sigma_{r}(\mathbf{w})) \\[0.2cm]
  & \quad \times \sigma_\hbar(\mathbf{w}/A_1,\delta/A_1) \cdots \sigma_\hbar(\mathbf{w}/A_r,\delta/A_r).
\end{align*}
We observe that
$$
\sigma_\hbar^1( \Sigma_1(\mathbf{w}), w_n + \Sigma_2(\mathbf{w}) + \cdots + \Sigma_r(\mathbf{w})) e^{- \frac{m_1 \sigma}{n-1} \vert \Sigma_1(\mathbf{w}) \vert} e^{- \frac{ m_2 \sigma}{n-1} \vert w_n + \Sigma_2(\mathbf{w}) + \cdots + \Sigma_r(\mathbf{w})) \vert)} \leq \frac{2(n-1)^2}{m_1 m_2 (e \sigma)^2},
$$
with $1 \leq m_1, m_2 \leq n-1$ to be determined. Recall that
\begin{align*}
\mathbf{c}(\delta) & = \mathbf{c}_{k_1, \ldots, k_r} \mathbf{c}(\delta/A_1) \cdots \mathbf{c}(\delta/A_r) = \frac{1}{n-1} \mathbf{c}_{k_2, \ldots, k_r} \mathbf{c}(\delta/A_1) \cdots \mathbf{c}(\delta/A_r),
\end{align*}
We denote $l_1 = \#A_1$, $l_2 = \# (A_2 \cup \cdots \cup A_r)$. If $l_1 = 1$ and $l_2 = n-2$, by the induction hypothesis we have:
\begin{align*}
\mathbf{c}(\delta) \sup_{\mathbf{w} \in (\mathcal{Z}^d)^n}  \Big \lbrace \vert \sigma_\hbar(\mathbf{w},\delta) \vert e^{-\sigma (\vert w_1  + \cdots +  w_n \vert)} \Big \rbrace & \\[0.2cm]
& \hspace*{-6.2cm} \leq 2C^{n-2} \left( \frac{ (n-1)}{e \sigma} \right)^{2(n-1)} \cdot \frac{ (n-2)^{2(n-2)}}{(n-1)m_1 m_2 (n-m_2 - 1)^{2(n-2)}(  (n-2)!)^2}.
\end{align*}
Taking $m_1 = n-1$, and $m_2 = 1$, the claim holds in this case. The case $l_2 = 0$ follows by same arguments. On the other hand, if $l_1 \geq 2$ and $l_2 \geq 1$, by the induction hypothesis we get:

\begin{align*}
\mathbf{c}(\delta) \sup_{\mathbf{w} \in (\mathcal{Z}^d)^n}  \Big \lbrace \vert \sigma_\hbar(\mathbf{w},\delta) \vert e^{-\sigma (\vert w_1  + \cdots +  w_n \vert)} \Big \rbrace & \\[0.2cm]
& \hspace*{-6.2cm} \leq 2C^{n-2} \left( \frac{ (n-1)}{e \sigma} \right)^{2(n-1)} \cdot \frac{(l_1 - 1)^{2(l_1-1)} l_2^{2l_2}}{(n-1)m_1 m_2 (n-m_1-1)^{2(l_1-1)} (n-m_2 - 1)^{2l_2}( (l_1-1)! l_2 !)^2}.
\end{align*}
We next use the following inequality for the Beta function \cite[Ex. 45, p. 263]{Whi_Wat96}:
\begin{equation}
\label{e:beta_function_estimate}
B(x,y) > \sqrt{2\pi} \frac{x^{x-\frac{1}{2}} y^{y-\frac{1}{2}}}{(x+y)^{x+y-1}}, \quad x,y \geq1,
\end{equation}
with $x = l_1$ and $y = l_2 +1$. We obtain:
\begin{align*}
\mathbf{c}(\delta) \sup_{\mathbf{w} \in (\mathcal{Z}^d)^n}  \Big \lbrace \vert \sigma_\hbar(\mathbf{w},\delta) \vert e^{-\sigma (\vert w_1  + \cdots +  w_n \vert)} \Big \rbrace & \\[0.2cm]
& \hspace*{-6cm} \leq  \left( \frac{ (n-1)}{e \sigma} \right)^{2(n-1)} \frac{1}{((n-1)!)^2} \cdot \frac{2 C^{n-2} n^{2n-1}}{(n-1)m_1 m_2 l_1 l_2 (n-m_1-1)^{2(l_1-1)}(n-m_2-1)^{2l_2}}.
\end{align*}
Choosing $m_1 = \frac{n-1}{l_1}$ and $m_2 = \frac{n-1}{l_2}$, the claim holds by taking $C$ sufficiently large.
\end{proof}

\begin{proof}[Proof of Theorem \ref{t:bound_lindstedt_series}]

Putting together Theorem \ref{t:main_lindstedt_teor}, Lemma \ref{l:lemma_eliasson}, and Proposition \ref{p:combinatorial_lemma}, we get:
$$
\Vert H_n \Vert_{s_0 - \sigma} \leq C^n \sum_{\delta \in \Delta(n)} \int_{(\mathcal{Z}^{d})^n }\prod_{v_j \neq 0} \vert v_j \vert^{3 \gamma} \vert  \mathbf{c}(\delta) \sigma_\hbar(\mathbf{w},\delta) \vert \vert \widehat{V}(\mathbf{w}) \vert e^{(s_0 - \sigma) \vert \Sigma(\mathbf{w}) \vert}  \kappa(d\mathbf{w}),
$$
where $\widehat{V}(\mathbf{w}) = \widehat{V}(w_1) \cdots \widehat{V}(w_n)$. Then, using Lemma \ref{e:analytical_part}, Lemma \ref{l:counting_trees}, and \eqref{e:elementary_estimate}, we obtain the claim by taking $s = s_0 - \sigma$ for suitable choice of $\sigma > 0$. The proof for the estimate of the counterterm $R'_n$ follows the same arguments but using \eqref{e:bound_omega_2} instead of \eqref{e:bound_omega_1}.

\end{proof}

\begin{remark}
\label{r:simplification}
Notice that in the case $\hbar = 1$ the term $\mathbf{c}(\delta) \sigma_\hbar(\mathbf{w},\delta)$ can be estimated by the trivial bound  $\vert \mathbf{c}(\delta)\sigma_\hbar(\mathbf{w},\delta) \vert \leq 2^{n-1}$.
This simplifies the proof in this case (and similarly in the case $\epsilon_\hbar = \epsilon \hbar$)  and allows to remove the analytic hypothesis on the $\xi$ variable.
\end{remark}

\section{Quantum renormalization and semiclassical measures} 
\label{s:main_proofs}

In this section we prove Theorems \ref{c:1}, \ref{c:2} and \ref{t:classic_main_theorem}.

\subsection{Proof of Theorems \ref{c:1} and \ref{t:classic_main_theorem}.} 

\begin{proof}[Proof of Theorem \ref{c:1}]

We start from \eqref{e:main_equation_for_lindstedt}. By Theorem \ref{t:bound_lindstedt_series} there exist solutions  $H \in \mathcal{A}_{s,\epsilon}(T^* \mathbb{T}^d)$ and $R' \in \mathcal{A}_{s,\epsilon}(\R^d)$ for sufficiently small $\epsilon = \epsilon(V,\omega) > 0$. This gives
$$
\frac{d}{dt} \Big( U_{-H}(t)^* \widehat{L}_{\omega,\hbar} U_{-H}(t) \Big) = \frac{i}{\hbar}U_{-H}(t)^* [\Op_\hbar(-H), \widehat{L}_{\omega,\hbar}] U_{-H}(t) = V - R'(t).
$$
Integrating this equation on the interval $[0,t]$ and defining $R(t) = \int_0^t R'(\tau) d\tau$, we get:
$$
U_{-H}(t)^* \widehat{L}_{\omega,\hbar} U_{-H}(t) = \widehat{L}_{\omega,\hbar} + tV - R(t).
$$
Defining $\mathcal{U}_\hbar(t) = U_{-H}(t)^*$, we obtain finally that
$$
\widehat{L}_{\omega,\hbar} = \mathcal{U}_\hbar(t)^* \Big( \widehat{L}_{\omega,\hbar} + \Op_\hbar (t V - R(t)) \Big) \mathcal{U}_\hbar(t).
$$

\end{proof}

\begin{proof}[Proof of Theorem \ref{t:classic_main_theorem}]
The proof reduces to that of Theorem \ref{c:1} simply by changing the Moyal commutator $[\cdot, \cdot]_\hbar$ into the Poisson Bracket $ \{ \cdot , \cdot \}$, and the quantum propagator $U_{-H}(t)$ into the time-dependent Hamiltonian flow $\Phi_t^{-H}$, since Lemmas \ref{l:Jacobi_rule}, \ref{l:commutator_loss} and \ref{l:propagator_symbolic_lemma} remain valid in the classical setting (see Remarks \ref{r:remark_poisson} and \ref{r:remark_flow}). Finally, defining $\Phi_t := (\Phi_t^{-H})^{-1}$, using that
$$
\Phi_t(z) - z =  \int_0^t X_{(\Phi_\tau^{-1})^* H} \circ \Phi_\tau(z) \,d\tau,
$$
where $X_H = \Omega \nabla H$ denotes the Hamitlonian vector field of $H$ and $\Omega$ stands for the canonical symplectic matrix on $\R^{2d}$, and using Lemma \ref{l:propagator_symbolic_lemma} (notice that the same proof applies for $\Phi_\tau$ and $\Phi_\tau^{-1}$), we obtain that
\begin{equation}
\label{e:estimate_flow}
\Vert \Phi_t - \operatorname{Id} \Vert_s \leq C \epsilon,
\end{equation}
for some  $0< s < s_0$ and all $0 \leq t \leq \epsilon$.
\end{proof}
 
\subsection{Proof of Theorem \ref{c:2}.}

Let $(\Psi_\hbar, \lambda_\hbar)$ be a sequence such that $\Vert \Psi_\hbar \Vert_{L^2(\mathbb{T}^d)} = 1$, $\lambda_\hbar \to 1$ as $\hbar \to 0^+$ and
$$
\Big( \widehat{L}_{\omega,\hbar} + \Op_\hbar (t V - R(t)) \Big) \Psi_\hbar = \lambda_\hbar \Psi_\hbar.
$$
Then, by Theorem \ref{c:1}, $\mathcal{U}_\hbar(t)^* \Psi_\hbar$ is a normalized sequence in $L^2(\mathbb{T}^d)$ and satisfies
$$
\big( \widehat{L}_{\omega,\hbar} - \lambda_\hbar \big) \, \mathcal{U}_\hbar(t)^* \Psi_\hbar = 0.
$$
Then, for any $a \in \mathcal{C}_c^\infty(T^* \mathbb{T}^d)$, we have:
$$
\big \langle \Op_\hbar(a) \Psi_\hbar, \Psi_\hbar \big \rangle_{L^2(T^* \mathbb{T}^d)} = \big \langle \mathcal{U}_\hbar(t)^* \Op_\hbar(a) \, \mathcal{U}_\hbar(t) \, \mathcal{U}_\hbar(t)^* \Psi_\hbar,  \mathcal{U}_\hbar(t)^* \Psi_\hbar \big \rangle_{L^2(T^* \mathbb{T}^d)}.
$$
Finally, using Egorov's Theorem (see for instance \cite[Thm 11.1]{Zw12}), we get
$$
\big \langle \Op_\hbar(a) \Psi_\hbar, \Psi_\hbar \big \rangle_{L^2(T^* \mathbb{T}^d)} = \big \langle \Op_\hbar(a \circ (\Phi_t^{-H})^{-1}) \, \mathcal{U}_\hbar(t)^* \Psi_\hbar, \mathcal{U}_\hbar(t)^* \Psi_\hbar \big \rangle_{L^2(T^* \mathbb{T}^d)},
$$
where $\Phi_t^{-H}$ is the classical flow generated by the Hamiltonian $-H$. Defining $\Phi_t := (\Phi_t^{-H})^{-1}$ the claim reduces to \cite[Prop.1]{Ar20} (notice that here we have defined semiclassical measures as probability measures). Finally, \eqref{e:estimate_flow} holds by the proof of Theorem \ref{t:classic_main_theorem}.

\section{Comparison with Eliasson's classical proof}
\label{s:comparison_eliasson}

In this section we compare the Lindstedt series given by Theorem \ref{t:main_lindstedt_teor} with the original idea of \cite{Elia89} to solve the classical problem. We write $z = (x,\xi) \in T^*\mathbb{T}^d$ and, c.f. \cite{Elia89}, we look for a symplectic map $\Phi_t$ written in the form $\Phi_t(z) = z + \mathcal{Z}_t(z)$ such that
\begin{equation}
\label{e:conjugation}
(\mathcal{L}_\omega + tV - R(t))\big \vert_{z + \mathcal{Z}_t(z)} = \mathcal{L}_\omega,
\end{equation} 
with unknowns $\mathcal{Z}_t(z)$ and $R(t)$.
Taking the symplectic gradient $\Omega \nabla$ at both sides of \eqref{e:conjugation}, we get
\begin{equation}
\label{e:differentiated_equation}
\Omega \nabla \mathcal{L}_\omega = \Omega ( \operatorname{Id} + d \mathcal{Z}_t(z))^T  \nabla (\mathcal{L}_\omega + tV-R(t))  \big \vert_{z + \mathcal{Z}_t(z)}.
\end{equation}
Then, using that the matrix $(\operatorname{Id} + d \mathcal{Z}_t(z))$ is symplectic (by assumption), we take its inverse at both sides of \eqref{e:differentiated_equation} to obtain
\begin{equation}
\label{e:classic_key_equation}
\omega \cdot  \nabla_x \mathcal{Z}_t(z) =  \Omega \nabla (tV-R(t))  \big \vert_{z + \mathcal{Z}_t(z)}.
\end{equation}

Expanding this equation in powers of $t$ we obtain recursive cohomological equations for $\mathcal{Z}_t$ and $R(t)$. Denoting $\mathcal{V} = \Omega \nabla V$, $\mathcal{R} = \Omega \nabla R$, and writing formally
$$
\mathcal{Z}_t = \sum_{n=1}^\infty t^n Z_n, \quad \mathcal{R}(t) = \sum_{n=1}^\infty t^n \mathcal{R}_n,
$$
we find (compare with \eqref{e:cohomological_general_equation}):
\begin{align}
\label{e:general_cohomological_classic}
\omega \cdot  \nabla_x Z_n & = \sum_{i_1 + \cdots + i_j = n} \frac{ \nabla_z^j (\mathcal{V} - \mathcal{R}_1) }{j!} [Z_{i_1}, \ldots, Z_{i_j}] - \sum_{i_1 + \cdots + i_k + k = n} \frac{\nabla_z^k \mathcal{R}_k }{k!} [ Z_{i_1}, \ldots, Z_{i_k}].
\end{align} 
From this equation we obtain the following classical Lindstedt series:
\begin{teor}
\label{t:main_lindstedt_teor_classic} For every $n \geq 1$, the solution to the cohomological equation \eqref{e:cohomological_general_equation} is given by:
\begin{align}
Z_{n}(x,\xi) & = \sum_{ \substack{ v \in (\mathbb{Z}^d)^n \\[0.1cm] \delta \in \Delta(n)}}  \int_{(\R^d)^n} \Omega_1(\delta,v) \widehat{\mathcal{V}}(\delta,v,\eta) e^{i(v_1 + \cdots + v_n) \cdot x} e^{i (\eta_1 + \cdots + \eta_n) \cdot \xi} d\eta, \\[0.2cm]
\mathcal{R}_n(\xi) & = \sum_{ \substack{ v \in (\mathbb{Z}^d)^n \\[0.1cm] \delta \in \Delta(n)}} \int_{(\R^d)^n} \Omega_2(\delta,v) \widehat{\mathcal{V}}(\delta,v,\eta) e^{i(\eta_1 + \cdots + \eta_n) \cdot \xi} d\eta,
\end{align}
where 
\begin{equation}
\label{e:tree_derivative_classic}
\widehat{\mathcal{V}}(\delta,v,\eta) =  \frac{ \widehat{\nabla_z^r \mathcal{V}}(v,\eta)}{r!}\left[ \widehat{\mathcal{V}}(\delta/A_1,v,\eta), \cdots, \widehat{\mathcal{V}}(\delta/A_r,v,\eta)\right],
\end{equation}
and $A \setminus \{n \} = A_1 \cup \cdots \cup A_r$ is the natural decomposition into simple index sets.
\end{teor}
From Theorem \ref{t:main_lindstedt_teor_classic} one can also show Theorem \ref{t:classic_main_theorem}. Moreover, Lemma \ref{e:analytical_part} can be replaced in this case by an estimate of \eqref{e:tree_derivative_classic}, which is easier since the coefficients $\mathbf{c}(\delta)$ simplify in this case by Lemma \ref{l:break_order} and symmetry of the multilinear map $\nabla_z^r \mathcal{V}$. However, it is not clear to the author how to quantize \eqref{e:conjugation} in an exact way. Alternatively, we consider more convinient to adopt the Hamiltonian formalism of Section \ref{s:hamiltonians}. This is the reason to search $\Phi_t$ (and hence $\mathcal{U}_\hbar(t)$), not as a general (quantizable) symplectic transformation but as a time-dependent Hamiltonian flow.

\appendix

\section{Pseudodifferential calculus on the torus}
\label{s:tools_of_analytic_calculus}

We include some basic lemmas on quantization of symbols in the spaces $\mathcal{A}_{s}(T^*\mathbb{T}^d)$, $\mathcal{A}_s(\R^d)$, $\mathcal{A}_s(\mathbb{T}^d)$. We fix $s > 0$ all along this appendix.

 \begin{definition}
 \label{d:semiclassical:quantization} Let $a : T^*\mathbb{T}^d \to \mathbb{C}$ be a symbol. The semiclassical Weyl quantization $\Op_\hbar(a)$ acting on $\psi \in \mathscr{S}(\mathbb{T}^d)$ is defined by
$$
\Op_\hbar(a) \psi(x) = \sum_{k,j\in\mathbb{Z}^d} \widehat{a}\left( k-j , \frac{ \hbar (k+j)}{2} \right) \widehat{\psi}(k) e^{ij\cdot x},
$$
where $\widehat{a}(k,\cdot)$ denotes the $(k)^{th}$-Fourier coefficient in the variable $x$.
\end{definition}

\begin{lemma}[Analytic Calderón-Vaillancourt theorem]
\label{l:analytic_calderon_vaillancourt_torus}
For every $a\in \mathcal{A}_s(T^*\mathbb{T}^d)$, the following holds:
\begin{equation}
\label{e:calderon_vaillancourt}
\Vert \Op_\hbar(a) \Vert_{\mathcal{L}(L^2(\mathbb{T}^d))} \leq C_{d,s} \Vert a \Vert_{s}, \quad \hbar \in (0,1].
\end{equation}
\end{lemma}
\begin{proof}
By the usual Calderón-Vaillancourt theorem, see for instance \cite[Prop 3.5]{Pau12}, the following estimate holds:
$$
\Vert \Op_\hbar(a) \Vert_{\mathcal{L}(L^2)} \leq C_d \sum_{\vert \alpha \vert \leq N_d} \Vert \partial_x^\alpha a \Vert_{L^\infty(T^* \mathbb{T}^d)}, \quad \hbar \in (0,1].
$$
Now, using the elementary estimate
\begin{equation}
\label{e:elementary_estimate}
\sup_{t \geq 0} t^m e^{-ts} = \left( \frac{m}{e s} \right)^m, \quad m > 0,
\end{equation}
we obtain
$$
\Vert \partial_x^\alpha a \Vert_{L^\infty(T^* \mathbb{T}^d)} \leq \frac{1}{(2\pi)^{d/2}} \sum_{k \in \mathbb{Z}^d} \vert k^\alpha \vert \Vert \widehat{a}(k, \cdot) \Vert_{L^\infty(\R^d)} \leq \left( \frac{\vert \alpha \vert}{e s} \right)^{\vert \alpha \vert} \Vert a \Vert_{s} = C_{\alpha, s} \Vert a \Vert_{s}.
$$
\end{proof}

Let $a,b :T^*\mathbb{T}^s \to \mathbb{C}$, the operator given by the composition $\Op_\hbar(a) \Op_\hbar(b)$ is another Weyl pseudodifferential operator with symbol $c$ given by the Moyal product $c = a \sharp_\hbar b$, see for instance \cite[Chp. 7]{Dim_Sjo99}. To write $c$ conveniently, we consider the product space $\mathcal{Z}^d := \mathbb{Z}^d \times \R^d$ and the measure $\kappa$ on $\mathcal{Z}^d$ defined by
\begin{equation}
\label{e:lebesgue_stieltjes}
\kappa(w) = \mathcal{K}_{\mathbb{Z}^d}(k) \otimes \mathcal{L}_{\R^d}(\eta), \quad w =  (k,\eta) \in \mathcal{Z}^d,
\end{equation}
where $\mathcal{L}_{\R^d}$ denotes the Lebesgue measure on $\R^d$, and
$$
\mathcal{K}_{\mathbb{Z}^d}(k) := \sum_{j \in \mathbb{Z}^d} \delta(k-j), \quad k \in \mathbb{Z}^d.
$$
Using this measure, we can write any function $a \in \mathcal{A}_{s,\rho}(T^*\mathbb{T}^d)$ as
\begin{equation}
\label{e:two_fourier}
a(z) = \frac{1}{(2\pi)^d} \int_{\mathcal{Z}^d} \mathcal{F} a(w) e^{iz \cdot w} \kappa(dw),
\end{equation}
where $z = (x,\xi) \in T^*\mathbb{T}^d$, and $\mathcal{F}$ denotes the Fourier transform in $T^*\mathbb{T}^d$:
\begin{equation}
\label{e:fourier_transform}
\mathcal{F} a(w) = \frac{1}{(2\pi)^d} \int_{T^*\mathbb{T}^d}  a(z) e^{-i w\cdot z} dz.
\end{equation}
With these conventions, the Moyal product $c = a \sharp_\hbar b$ can be written by the following  integral formula:
\begin{align}
\label{e:moyal_product}
a \sharp_\hbar b(z) = \frac{1}{(2\pi)^{2d}}\int_{\mathcal{Z}^{d} \times \mathcal{Z}^d} \big( \mathcal{F} a \big) (w') \big( \mathcal{F} b \big)(w-w') e^{\frac{i \hbar}{2} \{ w', w-w'\}} e^{i z \cdot w} \kappa(dw') \kappa(dw),
\end{align}
where $\{\cdot , \cdot  \}$ stands for the standard symplectic product in $\mathcal{Z}^d \times \mathcal{Z}^d$:
$$
\{w,w' \} = k \cdot \eta' - k' \cdot \eta, \quad w = (k,\eta), \quad w' = (k',\eta').
$$
Alternatively, we can deduce from \eqref{e:moyal_product} the following formula:
\begin{equation}
\label{e:moyal_product_2}
a \sharp_\hbar b(x,\xi) = \frac{1}{(2\pi)^d} \sum_{k,k'\in \mathbb{Z}^d} \widehat{a} \left(k', \xi + \frac{\hbar(k-k')}{2} \right) \widehat{b} \left( k-k', \xi - \frac{ \hbar k'}{2} \right) e^{i k \cdot x}.
\end{equation}
We will employ the notation 
\begin{equation}
\label{e:commutator_symbol}
[a,b]_\hbar := \frac{i}{\hbar} ( a \sharp_\hbar b - b \sharp_\hbar a),
\end{equation}
for the Moyal commutator. Hence $[\Op_\hbar(a), \Op_\hbar(b)] = \Op_\hbar([a,b]_\hbar)$. This immediately implies:
\begin{lemma}[Jacobi identity]
\label{l:Jacobi_rule}
Let $a,b,c \in \mathcal{A}_s(T^*\mathbb{T}^d)$. The following holds:
\begin{equation}
\label{e:jacobi_rule}
[a, [b,c]_\hbar ]_\hbar = [[a,b]_\hbar, c]_\hbar + [b, [ a, c]_\hbar]_\hbar.
\end{equation}
\end{lemma}
\begin{proof}
It is sufficient to observe that, for every $w_1,w_2,w_3 \in \mathcal{Z}^{2d}$,
\begin{align}
\label{e:jacobi_in_fourier}
 \sigma_\hbar^1(w_1,w_2+w_3)\sigma_\hbar^1(w_2,w_3) = \sigma_\hbar^1(w_1+w_2,w_3)\sigma_\hbar^1(w_1,w_2)  + \sigma_\hbar^1(w_2,w_1+w_3)\sigma_\hbar^1(w_1,w_3).
 \end{align}
\end{proof}
\begin{remark}
\label{r:jacobi_interpretation}
Notice that the Jacobi identity \eqref{e:jacobi_rule} can be understood as a chain rule of derivation.
\end{remark}
By definition \eqref{e:norm} of the norm of the space $\mathcal{A}_s(T^* \mathbb{T}^d)$, we have:
\begin{equation}
\label{e:compact_norm}
\Vert a \Vert_{s} = \frac{1}{(2\pi)^d} \int_{\mathcal{Z}^d} \vert \mathcal{F}a(w) \vert e^{\vert w \vert s} \kappa(dw).
\end{equation}
\begin{lemma}
\label{l:commutator_loss}
Let $a,b \in \mathcal{A}_{s}(T^*\mathbb{T}^d)$. Then, for $0 < \sigma_1 + \sigma_2 < s$:
\begin{equation}
\label{e:first_commutator}
\big \Vert [a,b]_\hbar \big \Vert_{s-\sigma_1-\sigma_2} \leq \frac{2}{e^2 \sigma_1(\sigma_1 + \sigma_2)}  \Vert a \Vert_{s} \Vert b \Vert_{s-\sigma_2}. 
\end{equation}
\end{lemma}

\begin{remark}
\label{r:remark_poisson}
The same estimate holds for the classical Poisson bracket $\{ \cdot, \cdot \}$ replacing $[\cdot, \cdot ]_\hbar$.
\end{remark}

\begin{proof}
By \eqref{e:moyal_product}, we have
\begin{equation}
\label{e:symbol_commutator}
[a,b]_\hbar(z) = \frac{2i}{\hbar} \int_{\mathcal{Z}^{2d}} \mathcal{F}a(w') \mathcal{F}b(w-w') \sin \left( \frac{\hbar}{2} \{ w', w-w' \} \right) \frac{e^{iw\cdot z}}{(2\pi)^{2d}} \kappa(dw') \, \kappa(dw).
\end{equation}
Then, using that
\begin{equation}
\label{e:bound_symplecti_product}
\vert \{w', w-w' \} \vert \leq 2 \vert w' \vert \vert w- w' \vert,
\end{equation}
we obtain:
\begin{align*}
\big \Vert [a,b]_\hbar \big \Vert_{s-\sigma_1-\sigma_2} & \\[0.2cm]
 & \hspace*{-2cm} \leq 2 \int_{\mathcal{Z}^{2d}} \vert \mathcal{F}a(w') \vert \vert w' \vert \vert \mathcal{F}b(w-w') \vert \vert w - w' \vert e^{(s-\sigma_1 - \sigma_2)(\vert w-w' \vert + \vert w' \vert)} \kappa(dw') \kappa(dw) \\[0.2cm] 
 & \hspace*{-2cm} \leq 2 \Big( \sup_{r\geq 0} r e^{-\sigma_1r} \Big) \Big( \sup_{r \geq 0} r e^{-(\sigma_1 + \sigma_2)r} \Big) \Vert a \Vert_{s} \Vert b \Vert_{s - \sigma_2} \\[0.2cm]
 & \hspace*{-2cm} \leq \frac{2}{e^2 \sigma_1 (\sigma_1 + \sigma_2)} \Vert a \Vert_{s} \Vert b \Vert_{s-\sigma_2}.
\end{align*}
\end{proof}

\begin{lemma}
\label{l:propagator_symbolic_lemma}
Let $\epsilon > 0$. Given $H(t) \in \mathcal{A}_{s,\epsilon}(T^*\mathbb{T}^d)$, let $U_H(t)$ be the unitary operator solving  \eqref{e:evolution_problem}. Let $0 < \sigma < s$. Assume that
\begin{equation}
\label{e:analytic_constrain}
\frac{2 \epsilon \Vert H \Vert_{s,\epsilon}}{\sigma^2} \leq \frac{1}{2}.
\end{equation}
Then, for every $a \in \mathcal{A}_s(T^*\mathbb{T}^d)$, there exists $C = C(H,\epsilon,\sigma) > 0$ and a symbol $\Psi_t^H(a) \in \mathcal{A}_{s-\sigma}(T^*\mathbb{T}^d)$ such that, for $0 \leq t  \leq \epsilon$,
\begin{equation}
\label{e:symbol_conjugation}
U_H^*(t) \Op_\hbar(a) U_H(t) = \Op_\hbar(\Psi_t^H(a)),
\end{equation}
and moreover:
$$
\Vert \Psi_t^H(a) \Vert_{s-\sigma} \leq 2 \Vert a \Vert_s.
$$
\end{lemma}

\begin{remark}
\label{r:remark_flow}
The same estimate holds for the classical flow $\Phi_t^H$ replacing $\Psi_t^H$.
\end{remark}

\begin{proof}
Using equations \eqref{e:evolution_problem} and \eqref{e:adjoint_evolution_equation}, we have formally that
$$
\frac{d}{dt} \Psi_t^H(a) = \Psi_t^H\big( [H(t), a]_\hbar \big).
$$
Expanding
$$
H(t) = \sum_{n=1}^\infty t^{n-1} H_n, \quad \Psi_t^{H}(a) = \sum_{n=0}^\infty t^n \psi_n(a),
$$
we find $\psi_0^{-1}(a) = 0$ and for $n \geq 1$, the recursive equation
$$
\psi_n(a) = \frac{1}{n} \sum_{j=0}^{n-1} \psi_j( [H_{n-j}, (a)]_\hbar).
$$
This gives:
$$
\psi_n(a) = \sum_{j=1}^n \sum_{k_1+ \cdots + k_j = n} \mathbf{c}_{k_j,\ldots,k_1}[ H_{k_j}, \cdots ,[ H_{k_1}, a]_\hbar \cdots ]_\hbar,
$$
where
$$
\mathbf{c}_{k_j,\ldots, k_1} := \frac{1}{k_1 + \cdots + k_j} \cdot \frac{1}{k_2 + \cdots + k_{j}} \cdots \frac{1}{k_j}.
$$
Using repeteadly \eqref{e:first_commutator} and Stirling's approximation, we get:
\begin{align*}
\Vert \psi_n(a) \Vert_{s-\sigma} & \leq \sum_{j=1}^n \sum_{k_1 + \cdots + k_j = n} \left(\frac{2}{\sigma^2}\right)^j \left( \frac{j^j}{e^j j!} \right)^2  j! \, \mathbf{c}_{k_j,\ldots,k_1} \Vert H_{k_1} \Vert_s \cdots \Vert H_{k_j} \Vert_s \Vert a \Vert_s \\[0.2cm]
 & \leq  \sum_{k_1 + \cdots + k_j = n} \left( \frac{2}{\sigma^2} \right)^j \Vert H_{k_1} \Vert_s \cdots \Vert H_{k_j} \Vert_s \Vert a \Vert_s .
\end{align*}
Summing up, we obtain:
$$
\Vert \Psi_t^H(a) \Vert_{s-\sigma} \leq \Vert a \Vert_s \sum_{j=0}^\infty \left( \frac{2 t \Vert H \Vert_{s,t}}{\sigma^2} \right)^j.
$$
Taking $t = \epsilon$ and using \eqref{e:analytic_constrain}, the claim holds.

\end{proof}

\subsection{Cohomological equations}
\label{s:cohomological_equations}

In this section we explain how to solve the cohomological equations appearing in our averaging method. This is a standard technique when dealing with small divisors problems.

\begin{lemma}
\label{l:solution_cohomological_equation}
Let $V \in \mathcal{A}_{s}(T^*\mathbb{T}^d)$. Then, the cohomological equation
\begin{equation}
\label{e:cohomological_equation_lemma}
\frac{i}{\hbar} [ \widehat{L}_{\omega,\hbar} , \Op_\hbar(F)] = \Op_\hbar(V - \langle V \rangle), \quad \langle F \rangle = 0,
\end{equation}
where
$$
\langle V \rangle(\xi) = \frac{1}{(2\pi)^d} \int_{\mathbb{T}^d} V(x,\xi) dx =  \frac{1}{(2\pi)^{d/2}} \widehat{V} (0,\xi),
$$
has a unique solution $F \in \mathcal{A}_{s-\sigma}(T^*\mathbb{T}^d)$, for every $0 < \sigma < s$, such that
\begin{equation}
\label{e:bound_solution_cohomologial}
\Vert F \Vert_{s-\sigma} \leq \varsigma^{-1} \left( \frac{\gamma}{e\sigma} \right)^{\gamma} \Vert V \Vert_{s}.
\end{equation}
\end{lemma}

\begin{proof}
Using the properties of the symbolic calculus for the Weyl quantization, which in this case is exact since $\mathcal{L}_\omega$ is a polynomial of degree one, equation \eqref{e:cohomological_equation_lemma} at symbol level is just
\begin{equation}
\label{e:symbolic_cohomological}
\{ \mathcal{L}_\omega, F \} = V - \langle V \rangle, \quad \langle F \rangle = 0.
\end{equation}
On the other hand, since
$$
\{ \mathcal{L}_\omega, F \}(x,\xi) =  \sum_{k \in \mathbb{Z}^d} i \omega \cdot k \,  \widehat{F}(k,\xi) e_k(x),
$$
we obtain the following formal series for the solution of \eqref{e:symbolic_cohomological}:
\begin{equation}
\label{e:solution_cohomological_equation}
F(x,\xi) =  \sum_{k \in \mathbb{Z}^d\setminus \{ 0 \}} \frac{\widehat{V}(k,\xi)}{i\omega \cdot k}  e_k(x).
\end{equation}
Finally, by the Diophantine condition \eqref{e:diophantine} and estimate \eqref{e:elementary_estimate}, we get \eqref{e:bound_solution_cohomologial}.
 Notice that the loss of analyticity of $F$ with respect to $V$ occurs only in the variable $x$.
\end{proof}

\section{Rooted planar trees and simple index sets}
\label{a:trees_and_index}

\begin{proof}[Proof of Lemma \ref{l:trees_and_oreders}] We proceed by induction. The case $n = 1$ is trivial. Let us assume that the first part of the statement is true for $n - 1 \geq 1$. Let $(A,\delta)$ be a  simple index set, we assume without loss of generality that $A = \{1, \ldots, n \}$. Let $a \in A$ be the unique element of $A$ such that $\delta(a) \geq 1$ and $\delta(b) = 0$ for all $1 \leq b < a$ (in particular $a \geq 2$). Then we design $B = \{ b \in A \, : \, 1 \leq b \leq \delta(a) \}$ to be the set of immediate predecessors of $a$. Now, let $A' = A \setminus   B$, we define the simple index set $(A',\delta')$ satisfying $\delta'(a') = \delta(a')$ if $a' \neq a$ and $\delta'(a) = 0$. By the induction hypothesis, there is a unique rooted planar tree $\mathcal{T}' = (A',\prec)$ satisfying
$$
\delta' = \upsilon_{\mathbf{d}'}(\mathcal{T}') \times \cdots \times \upsilon_0(\mathcal{T}'),
$$
where $\mathbf{d}' = \mathbf{d}(\mathcal{T}')$. Now, it is clear that $\mathcal{T} = (A,\prec)$ given by adding the elements of $B$ to $A'$ as immediate predecessors of $a$ defines a rooted planar tree. Let us define $\widetilde{\upsilon}_{\mathbf{d}}$ by
$$
\widetilde{\upsilon}^{\mathbf{d}'}_j := \left \lbrace \begin{array}{ll}
\upsilon^{\mathbf{d}'}_j(\mathcal{T}'), & \text{if } a_j \neq a, \\[0.2cm]
\delta(a), & \text{if } a_j = a.
\end{array} \right.
$$
Thus, \eqref{e:labelling} holds provided that $\upsilon_l(\mathcal{T}) = \upsilon_l(\mathcal{T}')$ for all $0 \leq l \leq \mathbf{d}(\mathcal{T}')-1$, and
\begin{align*}
\upsilon_{\mathbf{d}}(\mathcal{T}) & := \underbrace{(0, \ldots, 0)}_{\delta(a)} \times \widetilde{\upsilon}_{\mathbf{d}'}, \hspace*{2.6cm} \text {if } \mathbf{d} := \mathbf{d}(\mathcal{T}) = \mathbf{d}(\mathcal{T}'), \\[0.2cm]
\upsilon_{\mathbf{d}}(\mathcal{T}) & := \underbrace{(0, \ldots, 0)}_{\delta(a)}, \quad \upsilon_{\mathbf{d}-1}(\mathcal{T}) = \widetilde{\upsilon}_{\mathbf{d}'}, \quad \text{ if } \mathbf{d} := \mathbf{d}(\mathcal{T}) = \mathbf{d}(\mathcal{T}') + 1.
\end{align*}
Conversely, let $\mathcal{T} = (A, \prec)$ be a tree structure, let $\delta$ be the labelling given by \eqref{e:labelling}. Since $l(1) = \mathbf{d}$, we have that $\delta(1) = 0$. Moreover, let $A' := A \setminus \{1\}$, by the induction hypothesis the labelling $\delta'$ given by \eqref{e:labelling} for the subtree $\mathcal{T}' := (A', \prec)$ defines a simple index set. Moreover, $\delta'(a') = \delta(a')$ except at the immediate successor $a$ of $1$ in $(\mathcal{T},\prec)$, where $\delta'(a) = \delta(a) -1 $. Then 
$$
\sum_{1 \leq i \leq n } \delta_i = n-1.
$$ 
Moreover, for every $j > a$, we have
$$
\sum_{j \leq i \leq n} \delta_i = \sum_{j \leq i \leq n} \delta'_i \geq n - j + 1,
$$
while if $j \leq a$, then
$$
\sum_{j \leq i \leq n} \delta_i = 1 + \sum_{j \leq i \leq n} \delta'_i \geq 1 + n - j.
$$
This shows that $(\delta,A)$ defines a simple index set, and concludes the proof.
\end{proof}

\bibliographystyle{plain}
\bibliography{Referencias}

\end{document}